\documentclass[twoside,journal]{IEEEtran}
\usepackage{amsfonts}
\usepackage{amssymb}
\usepackage{amsmath}
\usepackage{mathrsfs}
\usepackage{verbatim}
\usepackage{subfigure}
\usepackage{balance}
\usepackage{booktabs}
\usepackage{fancybox}
\usepackage{bm}
\usepackage{extarrows}
\usepackage{algorithm}
\usepackage{algorithmic}
\usepackage{multirow}
\usepackage{array}
\usepackage{graphicx}
\usepackage{epstopdf}
\newtheorem{theorem}{Theorem}
\newtheorem{lemma}{Lemma}
\newtheorem{corollary}{Corollary}
\newtheorem{proof}{Proof}
\newtheorem{proposition}{Proposition}
\newtheorem{property}{Property}

\begin{document}

\title{Multi-Antenna Transmission in Downlink Heterogeneous Cellular Networks under A Threshold-based Mobile Association Policy}
\author{Tong-Xing~Zheng,~\IEEEmembership{Student Member,~IEEE,~}
        Hui-Ming~Wang,~\IEEEmembership{Senior~Member,~IEEE,}
and~Moon~Ho~Lee,~\IEEEmembership{Life~Senior~Member,~IEEE}
\thanks{T.-X. Zheng and H.-M. Wang are with the School of Electronic and Information Engineering, and also with the MOE Key Lab for Intelligent Networks and Network Security, Xi'an Jiaotong University, Xi'an, 710049, Shaanxi, China. Email: {\tt txzheng@stu.xjtu.edu.cn},
	{\tt xjbswhm@gmail.com}. }
\thanks{M. H. Lee is with the Division of Electronics Engineering, Chonbuk National 	University, Jeonju 561-756, Korea. Email: {\tt moonho@jbnu.ac.kr}.}
}

\maketitle
\vspace{-0.8 cm}

\begin{abstract}
	With the recent emergence of 5G era, heterogeneous cellular networks (HCNs) have invoked a popular research interest.
	In this paper, we provide a comprehensive analysis for multi-antenna transmissions in a multi-tier downlink HCN.
	We first propose a reliability-oriented threshold-based mobile association policy, where each user connects to the strongest base station from which this user can obtain the largest \emph{truncated long-term received power}.
	Under our mobile association policy, we derive analytical expressions for the exact outage probability of an arbitrary randomly located user, along with computationally convenient lower and upper bounds.
	Asymptotic analysis on the outage probability shows that introducing a large access threshold into mobile association significantly decreases the outage probability.
	We further investigate the spectrum efficiency and the energy efficiency of the HCN.
	Our theoretic analysis and numerical validations show that both the spectrum and energy efficiencies can be improved by properly choosing the access threshold.

\end{abstract}


\IEEEpeerreviewmaketitle

\section{Introduction}

\IEEEPARstart{H}{eterogeneous} cellular network (HCN) is a mature deployment of cellular networks in the LTE system embodying 4G \cite{Andrews2014What}.
Through densely deploying a variety of low-power infrastructure, such as pico and femto base stations (BSs)  over an existing macro cellular network, an HCN achieves seamless wireless coverage and high spectrum efficiency\footnote{An HCN is generally characterized as a multi-tier hierarchical architecture, where BSs belong to different tiers have different transmit powers and thus different coverage ranges.
	For example, a macrocell (in the first tier) has the highest power to provide the largest coverage, whereas a low-power femtocell simply acts as a home BS to realize the short-distance communication.}.
On the run for 5G, HCN is believed to continue to play a key role in pursuing a 1000x data rate increase.

\subsection{Related Work and Motivation}

In current HCN deployment, small cell infrastructure is often installed to meet specific demands, and its ``plug-and-play'' manner results in a dynamic topology.
Moreover, due to the built-out urban areas,
macrocells are becoming increasingly irregularly deployed and hence their locations become more and more random \cite{Heath2013Modeling}.
Traditional network models such as the overly simplified Wyner model and the completely regular square or hexagonal grid model   have failed to keep up with the network deployment trend nowadays.

Recently, tools from stochastic geometry theory have provided a mathematical framework to analyze random wireless networks by modeling the locations of network nodes as some spatial distributions, among which Poisson point process (PPP)  has been extensively used because of its extreme tractability \cite{Haenggi2009Stochastic}.
Andrews \emph{et al.} \cite{Andrews2011Tractable} have shown that the downlink signal-to-interference-plus-noise ratio (SINR) distribution derived based on the PPP model is about as accurate as the grid model.
Under the framework of stochastic geometry, network spectrum efficiency \cite{Jo2012Heterogeneous}-\cite{Dhungana2016Multichannel}, energy efficiency \cite{Quek2011Energy}, \cite{Liu2016Optimal} and security \cite{Zheng2014Transmission}-\cite{Zheng2016Optimal} have been investigated.
However, these early works are confined to either single-tier cellular networks or single-antenna heterogeneous networks.

As a natural extension, studies of the downlink HCNs have been carried out in multi-antenna scenarios recently \cite{Heath2013Modeling}, \cite{Dhillon2013Downlink}-\cite{Li2016Success}.
Specifically, Heath \emph{et al.} \cite{Heath2013Modeling} have investigated the interference distribution of a user that is associated with a fixed-size cell inscribed within a weighted Voronoi cell in a Poisson field of interferers.
Dhillon \emph{et al.} \cite{Dhillon2013Downlink} have derived closed-form expressions for both the coverage probability and the rate per user by using tools from stochastic orders, which however are not analytically tractable.
Adhikary \emph{et al.} \cite{Adhikary2015Massive} have proposed interference coordination strategies through spatial blanking by exploiting the directionality in channel vectors at the massive MIMO regime.
Li \emph{et al.} \cite{Li2016Success} have developed a semi-closed expression for the success probability using the Toeplitz matrix representation in a multiuser MIMO HCN, where the tradeoff between the link reliability and the area spectrum efficiency has been discussed.
Shojaeifard \emph{et al.} in a very recent contribution \cite{Shojaeifard2016Design} have characterized the spectral efficiency of a typical user in a multiuser MIMO HCN using a non-direct moment-generating-function-based methodology with closed-form expressions of both the received signal and the aggregate network statistics derived.
A fundamental challenge facing a multi-antenna HCN might be the computation of the SINR distribution, which unfortunately is still not well addressed by existing literature, and new network behaviors induced by the use of multiple antennas still need to shed light.

Another challenge in deploying an HCN lies in managing network interference.
Co-channel spectrum sharing between tiers makes the interference perceived at an arbitrary user in an HCN much more severe than that in a traditional cellular network, imperilling the successful co-existence of tiers \cite{Xia2010Open}.
Roche \emph{et al.} \cite{Roche2010Access} have shown that femtocell access control including closed and open access\footnote{Only the specified subscribers can use femtocell access points in closed access, whereas arbitrary nearby users can be served by femtocells in open access, which provides an inexpensive way to expand network capacity.}, is an efficient approach for interference management.
It is further reported in \cite{Xia2010Open} that compared to closed access, femtocell open access is preferred by network designers due to its superiority in reducing cross-tier interference.
Since each user is allowed to access an arbitrary tier, the issue of \emph{mobile association} has become a significant concern to balance cell load and quality of service (QoS).
Different from the nearest-BS criterion for single-tier cellular networks \cite{Andrews2011Tractable}, \cite{Li2014Throughput}, various mobile association policies have been proposed for HCNs, associating a mobile user to the strongest BS in terms of the \emph{instantaneous} SINR \cite{Dhillon2013Load}, \cite{Dhillon2012Modeling}, \emph{long-term} (as opposed to instantaneous) equivalent received power \cite{Wang2011Novel,Singh2013Offloading}, and \emph{long-term} biased received power \cite{Jo2012Heterogeneous,Li2016Success}, etc.
Further, the authors in \cite{Jo2012Heterogeneous} figure out that, in order to prevent the ``ping-pong effect'', i.e., unnecessary handovers caused by shadowing and fading, network designers typically prefer mobile association on the long-term received power.

However, existing mobile association policies, e.g., see \cite{Andrews2011Tractable}, \cite{Jo2012Heterogeneous}, \cite{Li2016Success}-\cite{Singh2013Offloading}, are confronted with the following two common problems, which motivate our research work.

\begin{itemize}
	\item
	The first problem is the loose constraint on user access.
	For example, a user whose distance from the strongest BS is very large might still be associated with this BS according to the given policy, regardless of the severe path loss.
	This actually results in poor end-user throughput, and also poor spectrum/energy efficiency. A proper way to overcome this problem is to forbidden these ``poor'' users to access, which not only improves the quality of transmission links but also frees resources, e.g., in time or space domain, to those ``better'' users\footnote{This will inevitably bring up the issue of coverage holes where ``poor'' users get no service. In order to cover for those coverage holes, we may 1) deploy additional infrastructure around a coverage hole or 2) increase the transmit power of BS or the number of transmit antennas to enlarge coverage range, but possibly at the cost of extra overhead, more complex networking, more severe interference or higher cost.}.
	\item
	The second problem is the controversial hypothesis that every BS is fully loaded, i.e., every BS has at least one user to serve at any time slot.
	This might be reasonable for a conventional macro cellular network where the BS density is far less than the user density, which however is not the case in an HCN where the BS density is comparable to the user density due to the cell densification, and thus a BS will remain idle when it has no user to serve.
	The probability of a BS being active is crucial to network performance since it determines the level of interference.
	Traditional methods in dealing with the BS active probability is either to parameterize  it \cite{Dhillon2013Load} or to calculate it empirically using for example data fitting \cite{Li2014Throughput,Singh2013Offloading}.
	A more precise way should be to compute the BS active probability based on the media access control (MAC) protocol adopted in the network, which however has not been well addressed by existing literature.
	
\end{itemize}

\subsection{Our Work and Contributions}
In this paper, we study the multi-antenna transmission in a $K$-tier HCN where positions of all tiers of BSs and users are modeled as independent homogeneous PPPs.
We provide a comprehensive analysis of the outage probability, spectrum and energy efficiencies under a stochastic geometry framework.
Our main contributions are summarized as follows:

1) We propose a novel mobile association policy based on the \emph{truncated long-term} received power (LTRP).
Specifically, a user connects to the BS that provides the largest LTRP, but remains inactive if the largest LTRP falls below a prescribed access threshold.
Under such a mobile association policy, we derive new closed-form expressions for the association probabilities of tiers and the BS activation probabilities.

2) We analyze the outage probability of a typical network user.
We provide not only an accurate integral expression for the outage probability in general scenarios, but also an analytical expression in the interference-limited case.
To further reduce the computation complexity and to develop more tractable insights into the relationship between the outage probability and network parameters, we derive new upper and lower bounds for the outage probability and execute some asymptotic analysis.
We show that the SIR \emph{invariance} property\footnote{If a user connects to the strongest BS, randomly adding new infrastructure does not influence the SIR distribution \cite{Jo2012Heterogeneous}, \cite{Dhillon2012Modeling}.} commonly observed in a single-antenna HCN gets lost in a multi-antenna case, and the introduced access threshold plays a critical role in decreasing the outage probability.

3) We further evaluate the spectrum and energy efficiencies of the HCN resorting to the asymptotic results on the outage probability.
We show that by properly choosing the access threshold, both spectrum and energy efficiencies can be significantly improved compared to a non-threshold policy.
In addition, although introducing more BSs (or users) and using more transmit antennas improve the spectrum efficiency, they may harm the energy efficiency due to the dramatic increase of transmit power and circuit power consumption.

\subsection{Organization and Notations}
The remainder of this paper is organized as follows.
In Section II, we describe the system model along with the mobile association policy.
In Section III, we investigate the outage probability of a randomly located user.
In Section IV, we evaluate the spectrum efficiency and energy efficiency of the HCN, respectively.
In Section V, we conclude our work.

\emph{Notations}:
bold uppercase (lowercase) letters denote matrices (column vectors).
$(\cdot)^{\mathrm{T}}$, $|\cdot|$, $\|\cdot\|$, $\mathbb{P}\{\cdot\}$, and $\mathbb{E}_A(\cdot)$ denote the transpose, absolute value, Euclidean norm, probability, and the expectation with respect to (w.r.t.) $A$, respectively.
$\mathcal{CN}(\mu, \sigma^2)$ denotes the circularly symmetric complex Gaussian distribution with mean $\mu$ and variance $\sigma^2$.
$\Gamma(N,\lambda)$ denotes the gamma distribution with parameters $N$ and $\lambda$.
$f_V(\cdot)$ and $F_V(\cdot)$ denote the probability density function (PDF) and cumulative distribution function (CDF) of a random variable $V$, respectively.
$\mathcal{B}(o,r)$ describes a disk with center $o$ and radius $r$.
$\mathcal{L}_{I}(s)$ denotes the 
Laplace transform of $I$ at $s$. 
${_2F_1}(\cdot)$ and $\gamma(\cdot,\cdot)$ denote the Gauss hypergeometric function and
incomplete gamma function, respectively.

\section{System Model}
This paper considers a downlink HCN consisting of $K$ tiers, where system parameters, e.g., BS transmit power, antenna number, biasing factor, etc., differ among tiers whereas remain unified in the same tier.
Denote the set of tier indexes by $\mathcal{K}\triangleq \{1,2,\cdots,K\}$.
In tier $k$, BSs are spatially distributed according to a homogeneous PPP $\Phi_k$ with density $\lambda_k$ in a 2-D plane.
User locations are also modeled as an independent homogeneous PPP $\Phi_u$ with density $\lambda_u$.
For convenience, we denote the BS at location $z$ in tier $k$ by $\mathrm{B}_z^{(k)}$ and the user at location $x$ by $\mathrm{U}_x$.

Wireless channels in the HCN are assumed to experience flat Rayleigh fading along with a large-scale path loss governed by the exponent $\alpha_k$ for tier $k$.
Each user is equipped with a single antenna, and the channel vector from $\mathrm{B}_z^{(k)}$ to $\mathrm{U}_x$ is expressed as $\mathbf{h}_{zx}r_{zx}^{-{\alpha_k}/{2}}$, where $\mathbf{h}_{zx}$ represents the small-scale fading vector with independent and identically distributed (i.i.d.) entries obeying $\mathcal{CN}(0,1)$, and $r_{zx}$ denotes the distance between $\mathrm{B}_z^{(k)}$ and $\mathrm{U}_x$.
We also assume that the thermal noise at $\mathrm{U}_x$ is $n_x\sim\mathcal{CN} (0,W)$.

To avoid the intra-cell interference, BSs are assumed to work in an orthogonal
transmission manner, e.g., time division multiple access (TDMA)\footnote{Generally speaking, TDMA may not be the optimal MAC protocol, and a higher network capacity can be achieved if a BS simultaneously serves multiple users in the same time-frequency resource, e.g., SDMA \cite{Li2016Success}, \cite{Ngo2013Energy}.
	However, TDMA is easily implemented in practice without requiring complex techniques, such as ``dirty-paper" coding \cite{Viswanath2003Sum} and yet efficient since intra-cell interference is avoided,
	In addition, as reported in \cite{Dhillon2013Downlink}, single-user beamforming provides higher coverages and per user data rates than SDMA does, which highlights the rationality of considering TDMA.}. Specifically, each BS allocates equal time slots to its associated users in a round-robin manner, and serves at most one user in a time slot. 
Furthermore, each BS only knows the channel state informations of its own users, and adopts the maximum ratio transmitting (MRT) beamforming to make the signal aggregated at the user side.
That is, when BS $\mathrm{B}_z^{(k)}$ serves user $\mathrm{U}_x$ at some point, the weighing vector at $\mathrm{B}_z^{(k)}$ is designed as the form of $\mathbf{w}_z=\frac{\mathbf{h}_{zx}}{\|\mathbf{h}_{zx}\|}$.

\subsection{Mobile Association Policy}

We consider an open-access HCN, in which each user is allowed to access any tier and connects to the BS that provides the largest LTRP \cite{Jo2012Heterogeneous,Li2016Success}.
The largest LTRP of a user related to tier $k$ is expressed as $\rho_k \triangleq {P_kM_kB_k}{D_k^{-\alpha_k}}$, where $P_k$, $M_k$, $B_k$ and $D_k$ denote the transmit power, the antenna number, the biasing factor, and the distance between this user and the nearest BS in tier $k$, respectively.
To guarantee a certain level of LTRP, we impose a threshold $\epsilon$ to mobile access, and allow only the users with their LTRPs beyond $\epsilon$ to be served.
In this way, the actual LTRP has the following \emph{truncated} form
\begin{align}\label{LTRP}
\hat{\rho}_k=\begin{cases}
~{P_kM_kB_k}{D_k^{-\alpha_k}}, & D_k\le R_k,\\
~0, & D_k>R_k,
\end{cases}
\end{align}
where $R_k=\left(\frac{P_kM_kB_k}{\epsilon}\right)
^{\frac{1}{\alpha_k}}$ represents the radius of the \emph{serving region}
$\mathcal{B}(o,R_k)$ of $\mathrm{B}_o^{(k)}$ \footnote{
	Since system parameters differ over tiers, the average coverage region of each cell in an HCN does not form a standard Voronoi tessellation, but closely resembles a multiplicatively weighted Voronoi diagram \cite{Dhillon2012Modeling}.}.

We stress that, the access threshold introduces in our mobile association policy triggers a non-trivial tradeoff between the end-to-end reliability and the area spectrum efficiency.

1) On one hand, setting a large access threshold guarantees good links to the served users, since the LTRP inside a BS's serving region is always superior to that outside.
A larger access threshold also reduces the number of active BSs and thus reduces the network interference.

2) On the other hand, setting a small access threshold activates more BSs and thus establishes more communication links per unit area, potentially improving the area spectrum efficiency although at the cost of degrading the link reliability.

In our mobile association policy, the index of the tier associating a typical user of interest is determined as
\begin{align}\label{k_opt}
k^* = \arg\max_{
	\begin{subarray}{c}
	k\in\mathcal{K}
	\end{subarray}
}\hat{\rho}_k,
\end{align}
and the association probability of tier $k$ is defined as
\begin{equation}\label{T_k_def}
\mathcal{T}_k\triangleq\mathbb{P}\{k^*=k\}=\mathbb{P}\{ \hat{\rho}_k>\hat{\rho}_j, \forall j\in\mathcal{K}\setminus k\}.
\end{equation}

For ease of notation, we let $\Omega_{k} \triangleq {P}_{k}{M}_{k}{B}_{k}$, $\delta_k\triangleq\frac{2}{\alpha_k}$, and
\begin{equation}\label{Notation}
{\mathcal{V}}_{j,k} \triangleq \frac{\mathcal{V}_j}{\mathcal{V}_k},~~ \forall ~ \mathcal{V}\in\{P,M,B,\Omega,\delta\}.
\end{equation}

\begin{lemma}\label{T_k_lemma}
	\textit{The association probability of tier $k$ is given by
		\begin{equation}\label{T_k_alpha}
		\mathcal{T}_k = \pi\lambda_k\int_0^{R^2_k}
		e^{-\pi\sum_{j=1}^{K}\lambda_j
			\Omega_{j,k}^{{\delta_{j}}}
			r^{{\delta_{j,k}}}}dr.
		\end{equation}
		When $\{\alpha_j\}=\alpha$, let $\Lambda\triangleq \sum_{j\in\mathcal{K}}\lambda_j\Omega_{j}^{{\delta}}$, and $\mathcal{T}_k$ is given by
		\begin{equation}\label{T_k}
		\mathcal{T}_k = \frac{\lambda_k\Omega^\delta_k}{\Lambda}
		\left(1-e^{-\pi\Lambda \epsilon^{-\delta}}\right).
		\end{equation}}
\end{lemma}
\begin{proof}
	Please see Appendix \ref{appendix_T_k_lemma}.
\end{proof}

From Lemma \ref{T_k_lemma}, we observe that if
tier $k$ has much larger $\lambda_k\Omega_k$ than other tiers do, it is more capable to access users, with probability $\mathcal{T}_k\approx1-e^{-\pi \lambda_k \Omega^{\delta}_k\epsilon^{-\delta}}$ which reaches one as $\lambda_k\Omega_k$ goes to infinity.
In addition, the introduced access threshold $\epsilon$ makes each user have probability  $1-\sum_{k\in\mathcal{K}}\mathcal{T}_k
=e^{-\pi\Lambda\epsilon^{-\delta}}$ of not being associated with any BSs.

Due to the restriction of the serving region (or the access threshold) and the competition for users, it is quite common for some BSs being idle\footnote{There are two situations where a BS will remain idle.
	The first one is that there is no user inside this BS's serving region.
	The second one is that all users within this BS's serving region are associated with the other BSs.}.
To precisely describe the behavior of mobile association and BS activation, we define the activation probability of an arbitrary BS in tier $k$ as
\begin{equation}\label{A_k_def}
\mathcal{A}_k \triangleq \mathbb{P}\{
\textit{A BS in tier k associates at least one user}\},
\end{equation}
with a closed-form expression given in the following lemma.
\begin{lemma}\label{A_k_lemma}
	\textit{The BS activation probability of tier $k$ is
		\begin{equation}\label{A_k_j}
		\mathcal{A}_k = 1-\exp\left(-\pi\lambda_u\int_0^{R^2_k}
		e^{-\pi\sum_{j=1}^{K}\lambda_j
			\Omega_{j,k}^{{\delta_{j}}}
			r^{{\delta_{j,k}}}}dr\right).
		\end{equation}
		When $\{\alpha_j\}=\alpha$, invoking $\Lambda$ defined in Lemma \ref{T_k_lemma} yields
		\begin{equation}\label{A_k}
		\mathcal{A}_k = 1-\exp\left(-\frac{\lambda_u\Omega^{\delta}_k}{\Lambda}
		\left(1-e^{-\pi\Lambda \epsilon^{-\delta}}\right)\right).
		\end{equation}}
\end{lemma}
\begin{proof}
	Please see Appendix \ref{appendix_A_k_lemma}.
\end{proof}

Note that $\mathcal{A}_k = 1-e^{-\frac{\lambda_u}{\lambda_k}\mathcal{T}_k}$, and the tier with larger $\Omega_k$ evidently has larger $\mathcal{A}_k$ than other tiers do.
We can also find that $\mathcal{A}_k$ increases with $\lambda_u$ whereas decreases with $\lambda_k$ (From \eqref{T_k}, we can easily prove $\frac{\lambda_u}{\lambda_k}\mathcal{T}_k$ monotonically decreases with $\lambda_k$).
This confirms the fact that introducing more users activates more BSs, while deploying more cells results in a larger possibility of BSs being idle due to the limited users.

As expected, introducing access threshold $\epsilon$ decreases $\mathcal{A}_k$, making the set of the active BSs in tier $k$ a thinning of $\Phi_k$, denoted by $\Phi_k^o$ with a new density $\lambda_k^{o}=\mathcal{A}_k\lambda_k$.
By calculating the derivative of $\lambda_k^{o}$ w.r.t. $\lambda_k$, we find that $\lambda_k^{o}$ monotonically increases with $\lambda_k$, although $\mathcal{A}_k$ decreases with $\lambda_k$.
Besides, since $\mathcal{T}_k\rightarrow 1$ as $\lambda_k\rightarrow\infty$, we have $\lim_{\lambda_k\rightarrow \infty}\lambda_k^{o} =\lim_{\lambda_k\rightarrow \infty}\lambda_k
\left(1-e^{-\frac{\lambda_u}{\lambda_k}}\right)
=\lambda_u$.
This confirms the fundamental truth that the number of active BSs should be always limited by that of users.
In this regard, those assumptions developed in \cite{Jo2012Heterogeneous}, \cite{Li2016Success}, \cite{Dhillon2012Modeling}-\cite{Singh2013Offloading}, in which all tiers of BSs are always fully loaded regardless of the number of users, are inappropriate, which inevitably result in an inaccurate evaluation of network performance, especially for a dense HCN.

\section{Outage Performance Analysis}
In this section, we investigate the outage probability of a randomly located user.
Without loss of generality, we consider a typical user $\mathrm{U}_o$.
When it is served by $\mathrm{B}_b^{(k)}$,
in addition to the desired signal from $\mathrm{B}_b^{(k)}$, $\mathrm{U}_o$ receives interferences both from the BSs in tier $k$ (except for the serving BS) and from the BSs in other tiers.
The outage probability can be mathematically defined as the likelihood of the event that the instantaneous SINR at $\mathrm{U}_o$ is lower than a target SINR $\beta$, i.e.,
\begin{equation}\label{O_k_def}
\mathcal{O}_{k}=\mathbb{P}\{\textsf{{SINR}}_{k}<\beta\},
\end{equation}
where $\textsf{{SINR}}_{k}$ has the form of
\begin{equation}\label{SINR}
\textsf{{SINR}}_{k} = \frac{P_k\|\mathbf{h}_{bo}\|^2 X_k^{-\alpha_k}}
{I_o+W}.
\end{equation}
Here, $I_o\triangleq \sum_{j\in\mathcal{K}}I_{jo}$ denotes the aggregate interference of the HCN with $I_{jo}\triangleq\sum_{z\in\Phi_j^o\setminus b} P_j{|\mathbf{h}_{zo}^{\mathrm{H}}\mathbf{w}_{z}|^2}
r^{-\alpha_j}_{zo}$, and
$X_k$ denotes the distance between $\mathrm{U}_o$ and $\mathrm{B}_b^{(k)}$.
Due to the mobile association policy \eqref{k_opt}, an exclusion region  $\mathcal{B}(o,r_{j,k})$ around $\mathrm{U}_o$ with radius $r_{j,k}=\Omega_{j,k}
^{{1}/{\alpha_j}}X_k^{{\alpha_k}/{\alpha_j}}$ exists, and all the interfering BSs in tier $j\in\mathcal{K}$ lie outside of this region.

Considering that each user is associated with only one tier, the overall outage probability of $\mathrm{U}_o$ can be obtained using the total probability formula  \cite{Jo2012Heterogeneous}, given by
\begin{equation}\label{O}
\mathcal{O} = \sum_{k\in\mathcal{K}}\mathcal{T}_k\mathcal{O}_{k},
\end{equation}
with $\mathcal{T}_k$ and $\mathcal{O}_{k}$ defined in \eqref{T_k_def} and \eqref{O_k_def}, respectively.

\subsection{General Results}
Define $s \triangleq \frac{\beta X_k^{\alpha_k}}{P_k}$.
$\mathcal{O}_{k}$ can be calculated as follows
\begin{align}\label{O_k1}
&\mathcal{O}_{k}
=\mathbb{E}_{X_k}\mathbb{E}_{I_o} \left[ \mathbb{P}\left\{\|\mathbf{h}_{bo}\|^2< {s}(I_o+W)\right\}\right]\nonumber\\
    &\!\!\stackrel{\mathrm{(a)}}
    = 1-\mathbb{E}_{X_k} \mathbb{E}_{I_o}\left[e^{-s(I_o+W)}
    \sum_{m=0}^{M_k-1}\frac{s^m(I_o+W)^m}
    {\Gamma(m+1)}\right]\nonumber\\
    &\!\!=1-\sum_{m=0}^{M_k-1}
    \mathbb{E}_{X_k}\mathbb{E}_{I_o} \left[e^{-sW}\sum_{n=0}^{m}\binom{m}{n}
    \frac{W^{m-n}s^me^{-sI_o}}
    {m!}I_o^n\right]\nonumber\\
   &\!\!\stackrel{\mathrm{(b)}}= 1- \sum_{m=0}^{M_k-1}\mathbb{E}_{X_k}\left[e^{-sW}
     \sum_{n=0}^{m}\binom{m}{n}\frac{(-1)^nW^{m-n}s^m}
    {m!}\mathcal{L}^{(n)}_{I_o}(s)\right],
\end{align}
where (a) holds for $\|\mathbf{h}_{bo}\|^2\sim \Gamma(M_k,1)$, and (b) holds for \cite[Theorem 1]{Hunter08Transmission} with $\mathcal{L}^{(n)}_{I_o}(s)$ the $n$-order derivative of $\mathcal{L}_{I_o}(s)$.

Generally speaking, it is difficult to express $\mathcal{L}^{(n)}_{I_o}(s)$ in an analytical form.
Fortunately, Li \emph{et al.} \cite{Li2014Throughput}, \cite{Li2016Success} have recently proposed a useful approach to handle $\mathcal{L}^{(n)}_{I_o}(s)$ for the multi-antenna cellular networks.
In the following, we continue to use Li's approach to calculate the outage probability under our threshold-based mobile association policy.
The key idea of the derivation of $\mathcal{O}_k$ is to first express $\mathcal{L}^{(n)}_{I_o}(s)$ in a linear recurrence form, and then average over $X_k$ within $(0, R_k]$.
The following theorem provides the most general result of $\mathcal{O}_k$.
\begin{theorem}\label{O_k_theorem}
\textit{The outage probability of $\mathrm{U}_o$ whenever it is associated with tier $k$ is given by
\begin{align}\label{O_k_j1}
\mathcal{O}_k&= 1 -\frac{\lambda_k}{\mathcal{T}_k} \sum_{i=0}^{M_k-1}\sum_{m=0}^{M_k-1}
    \sum_{n=0}^{m}\int_0^{R_k^2}\frac{\pi^{i+1}\sigma_k^{m-n}}
    {(m-n)!i!}x^{\frac{\alpha_k}{2}(m-n)}\times\nonumber\\
    & \!\!\! e^{-\sigma_k x^{\alpha_k/2}-\pi g_{k0}-\pi\sum_{j=1}^{K}\lambda_j
\Omega_{j,k}^{{\delta_{j}}}
x^{{\delta_{j,k}}}}\mathbf{Q}_{M_k}^i(n+1,1)dx.
\end{align}
$\mathbf{Q}_{M_k}^i(p,q)$ denotes the $(p,q)$-th entry of the $i$-power $\mathbf{Q}_{M_k}^i$, and $\mathbf{Q}_{M_k}$ is a lower triangular Toeplitz matrix given by
\begin{align}\label{QM}
\mathbf{Q}_{M_k} \triangleq \left[\begin{array}{ccc}
    0 & & \\
    g_{k,1} & 0  & \\
    g_{k,2} & g_{k,1} & 0  ~~~ ~~~  \\
    \vdots  & & ~~\ddots\\
    g_{k,M_k-1} & g_{k,M_k-2} & \cdots ~~g_{k,1} ~~ 0
\end{array}\right],
\end{align}
where $g_{k,n}= \sum_{j=1}^K\mathcal{A}_j
\lambda_j\Omega^{\delta_{j}}_{j,k}\varphi_{j,k}(n)
X_k^{2\delta_{j,k}}$ with
\begin{align}\label{varphi}
    \varphi_{j,k}(n) =
    \begin{cases}
        \frac{\delta_j\tau_{j,k}
        {_2}F_1\left(1,1-{\delta_j};2-{\delta_j},
-\tau_{j,k}\right)}{(1-\delta_j)}, & n = 0\\
 \frac{\delta_j\tau^n_{j,k}
 {_2}F_1\left(n+1,n-{\delta_j};n+1-{\delta_j},
-\tau_{j,k}\right)}
{(n-\delta_j)},& n\ne 0
    \end{cases}
\end{align}
and $\tau_{j,k}\triangleq \frac{\beta}{M_{j,k}B_{j,k}}$,
$\sigma_k\triangleq \frac{\beta_{k}W}{P_k}$.}

\textit{When $\{\alpha_j\}=\alpha$, we further obtain
\begin{align}\label{O_k}
     \mathcal{O}_{k}=1-
     \frac{\lambda_k}{\mathcal{T}_k}
     \sum_{i=0}^{M_k-1}\sum_{m=0}^{M_k-1}
     \sum_{n=0}^{m}
     \frac{\pi^{i+1}\sigma_k^{m-n}
     \mathcal{Q}_k}
     {i!(m-n)!} \mathbf{\Theta}_{M_k}^i(n+1,1),
\end{align}
where $\mathbf{\Theta}_{M_k}$ has the same form as $\mathbf{Q}_{M_k}$ does, simply with $g_{k,n}$ in \eqref{QM} replaced by $\theta_{k,n}= \sum_{j=1}^K\mathcal{A}_j
\lambda_j\Omega^{\delta_{j}}_{j,k}\varphi_{j,k}(n)$,
and
$\mathcal{Q}_{k}=\int_0^{R_k^2}
x^{i+\frac{\alpha}{2}(m-n)}e^{-\sigma_k x^{\frac{\alpha}{2}}-\pi\mathcal{Z}_k x}dx$,
with
$\mathcal{Z}_k=\theta_{k,0}
+\frac{\Lambda}{\Omega_k^{\delta}}$.}
\end{theorem}
\begin{proof}
    Please see Appendix \ref{appendix_O_k_theorem}.
\end{proof}

Due to the existence of the integral term $\mathcal{Q}_{k}$, \eqref{O_k} is rather unwieldy to analyze.
For the special case $\alpha=4$, $\mathcal{Q}_{k}$ can be simplified, which makes $\mathcal{O}_{k}$ semi-closed.
\begin{corollary}\label{O_k_a_4}
\textit{When $\alpha=4$, $\mathcal{Q}_{k}$ in \eqref{O_k} is given by
\begin{align}\label{Q}
  \mathcal{Q}_{k}&= \sum_{j=0}^{m-n+i}
  \binom{m-n+i}{j}\left(\frac{-\pi\mathcal{Z}_k}
  {2\sigma_k}  \right)^{m-n+i-j}
  \frac{\sigma_k^{-\frac{j+1}{2}}}{2}e^{
  \frac{\pi^2\mathcal{Z}^2_k}{4\sigma_k}}\nonumber\\
  &\left(\gamma\left(\frac{j+1}{2},
  \sigma_k\left(R_k^2+\frac{\pi\mathcal{Z}_k}
  {2\sigma_k}\right)^2
  \right)-\gamma\left(\frac{j+1}{2},
  \frac{\pi^2\mathcal{Z}^2_k}{4\sigma_k}\right)\right).
\end{align}}
\end{corollary}
\begin{proof}
    Transforming the integral variable $x$ in $\mathcal{Q}_{k}$ as $x+\frac{\pi\mathcal{Z}_k}{2\sigma_k}
    \rightarrow \sqrt{y}$, and after some algebraic operations combined with formula \cite[(3.381.1)]{Gradshteyn2007Table}, we complete the proof.
\end{proof}

Although the results given in Theorem \ref{O_k_theorem} and Corollary \ref{O_k_a_4} are too involved to analyze, the main contribution is to provide a general and meanwhile accurate expression for the outage probability without requiring time-consuming Monte Carlo simulations.
To facilitate the analysis, we are going to seek simpler expressions for $\mathcal{O}_{k}$ considering some special cases.

\subsection{Interference-limited HCN}
Ubiquitous interference in HCNs makes it reasonable for us to consider the interference-limited case by ignoring the thermal noise at the user side, i.e., $W=0$.
In this case, the outage probability $\mathcal{O}_{k}$ in \eqref{O_k1} can be simplified as
\begin{equation}\label{O_k_int_def}
  \mathcal{O}^{int}_{k}=1-
  \sum_{m=0}^{M_k-1}\mathbb{E}_{X_k}\left[
  \frac{(-1)^m s^m} {m!}
    \mathcal{L}^{(m)}_{I_o}(s)\right].
\end{equation}

\begin{theorem}\label{O_k_int_theorem}
\textit{For the interference-limited HCN with $W=0$, the outage probability of $\mathrm{U}_o$ associated with tier $k$ is given by
\begin{align}\label{O_k_int_j}
&\mathcal{O}_k^{int}= 1 - \frac{\lambda_k}{\mathcal{T}_k} \sum_{i=0}^{M_k-1}\sum_{m=0}^{M_k-1}\times\nonumber\\
     &\!\!\!\int_0^{R_k^2}\frac{\pi^{i+1}}
    {i!}e^{-\pi g_{k0}-\pi\sum_{j=1}^{K}\lambda_j
\Omega_{j,k}^{{\delta_{j}}}
x^{{\delta_{j,k}}}}\mathbf{Q}_{M_k}^i(m+1,1)dx.
\end{align}
When $\{\alpha_j\}=\alpha$, we further obtain \begin{align}\label{O_k_int}
 \mathcal{O}^{int}_{k}
     =1-\frac{\lambda_k}{\mathcal{T}_k}
     \sum_{i=0}^{M_k-1}
     \frac{\left\|\mathbf{\Theta}_{M_k}^i\right\|_1}
     {\mathcal{Z}^{i+1}_k}\left(
     1-\sum_{l=0}^i\frac{\pi^{l}
     e^{-\pi\mathcal{Z}_k R_k^2}} {l!~R_k^{-2l}\mathcal{Z}^{-l}_k}\right),
\end{align}
with $\|\mathbf{A}\|_1=\max_{1\le j\le n}\sum_{i=1}^m|A_{ij}|$ and $\mathbf{A}$ an $m\times n$ matrix.}
\end{theorem}
\begin{proof}
Substituting $W=0$ into  \eqref{pc_app_final} yields \eqref{O_k_int_j}.
When $\{\alpha_j\}=\alpha$, $\mathcal{O}^{int}_{k}$ 
can be further expressed as
\begin{equation}\label{pc_app_inter_final2}
\mathcal{O}^{int}_{k} = 1-\sum_{i=0}^{M_k-1}
\mathbb{E}_{X_k}\left[
\left\|\frac{\pi^i}{i!}
{y_0 X_k^{2i}}
\mathbf{\Theta}_{M_k}^i\right\|_1\right].
\end{equation}
Averaging over $X_k$ completes the proof.
\end{proof}

\begin{figure}[!t]
\centering
\includegraphics[width= 3.0in]{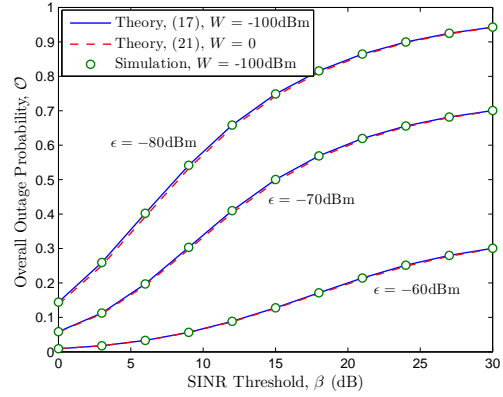}
\caption{Overall outage probability $\mathcal{O}$ vs. SINR threshold $\beta$ in a macro/pico/femto HCN.
$\{P_1,P_2,P_3\}=\{30,10,0\}$dBm, $\{M_1,M_2,M_3\}=\{4,2,1\}$, $\lambda_1=\frac{1}{\pi 500^2\mathrm{m}^2}$, $\{\lambda_2,\lambda_3,\lambda_u\}=\{4\lambda_1, 10\lambda_1,20\lambda_1\}$, $\{B_1,B_2,B_3\}=\{1,2,4\}$, and $\{\alpha_j\}=4$.}
\label{OP_W_BETA_EP}
\end{figure}

One can observe that $\mathcal{O}^{int}_{k}$ (i.e., the CDF of SIR) is impacted by $\{\lambda_j\}$, which means the SIR \emph{invariance} property observed in a single-antenna HCN \cite{Jo2012Heterogeneous}, \cite{Dhillon2012Modeling} may no longer hold in a multi-antenna case.
The underlying reason is that both the mobile association policy and the distributions of the signal and interference depend heavily on the number of the transmit antennas equipped at the BS side and the adopted
multi-antenna transmit strategy.  
The term
$\sum_{l=0}^i\frac{\pi^{l}
     e^{-\pi\mathcal{Z}_k R_k^2}} {l!R_k^{-2l}\mathcal{Z}^{-l}_k}$ in \eqref{O_k_int} arises from $\epsilon$, which is zero under a non-threshold policy (i.e., $\epsilon=0$) since $R_k^{\epsilon=0}\rightarrow\infty$, and the corresponding outage probability simplifies to
\begin{equation}\label{O_k_ep_0}
  \mathcal{O}^{int, \circ}_{k}
     =1-\frac{\Lambda}{\Omega_k^{\delta}}
     \sum_{i=0}^{M_k-1}
     \frac{1}{\mathcal{Z}^{i+1}_k}
     \left\|\mathbf{\Theta}_{M_k}^i\right\|_1.
\end{equation}

\begin{corollary}\label{Upsilon_k_corollary}
    \textit{For the special case of $\{M_j\}=M$, $\{B_j\}=B$, and $\{\alpha_j\}=4$,  $\mathcal{Z}_k$ in \eqref{O_k_int} can be further simplified to the following single trigonometric function
    \begin{align}\label{Upsilon_k}
       \mathcal{Z}_{k} = \sum_{j=1}^{K}\lambda_j\sqrt{P_{j,k}}
      \left(1+{\mathcal{A}_j\sqrt{\beta}}
      \arctan\sqrt{\beta}\right)   .
    \end{align}
    }
\end{corollary}
\begin{proof}
    When $\{M_j\}=M$, $\{B_j\}=B$, and $\{\alpha_j\}=4$, $\mathcal{Z}_k=\theta_{k,0}
    +{\Lambda}{\Omega^{-\delta}}$ can be rewritten as
    \begin{equation}\label{Z_k_4}
      \mathcal{Z}_k = \sum_{j=1}^{K}\lambda_j\sqrt{P_{j,k}}
      (1+\mathcal{A}_j\varphi_{j,k}(0)),
    \end{equation}
    where $\varphi_{j,k}(0)$ can be calculated from \eqref{varphi_app} as
    \begin{align}\label{varphi_j_k}
        \varphi_{j,k}(0) &= \frac{1}{2}\sqrt{\beta}\int_0^{\beta}
        \frac{u^{- \frac{1}{2}}}{1+u}du= \sqrt{\beta}\arctan\sqrt{\beta},
    \end{align}
    with the last equality holds for \cite[(2.211)]{Gradshteyn2007Table}.
    Substituting \eqref{varphi_j_k} into \eqref{Z_k_4} completes the proof.
\end{proof}

The comparison between the outage probabilities given in \eqref{O_k} and \eqref{O_k_int} is shown in Fig. \ref{OP_W_BETA_EP}.
 We see that the interference-limited result $\mathcal{O}^{int}_{k}$ in \eqref{O_k_int} approximates the exact one in \eqref{O_k}.
Hereafter, we focus on $\mathcal{O}^{int}_{k}$ for convenience.
As analyzed in Sec. II-A, increasing $\epsilon$ decreases $\mathcal{O}^{int}_{k}$.
That is to say, our threshold-based mobile access policy is beneficial to the end-to-end reliability.

\subsection{Upper and Lower Bounds}
Although \eqref{O_k_int} gives an analytical expression for the outage probability $\mathcal{O}^{int}_{k}$, it brings a heavy computation burden due to the matrix exponential operation $\mathbf{\Theta}_{M}^i$, especially when the value of $M$ or $i$ is large.
We observe that the term $\mathbf{\Theta}_{M}^i$ is caused by the $i$-order derivative $\mathcal{L}^{(i)}_{I_o}(s)$, which can be avoided by invoking Alzer's inequality \cite{Alzer1997Mathmatics}.
In the following theorem, we provide computationally convenient upper and lower bounds of $\mathcal{O}^{int}_{k}$.

\begin{theorem}\label{O_k_bound_theorem}
\textit{For the interference-limited HCN, the outage probability of a typical user associated with tier $k$ satisfies
\begin{equation}\label{pc_theorem2}
  \mathcal{O}^{int,L}_{k}\le\mathcal{O}^{int}_{k}
  \le\mathcal{O}^{int,U}_{k}.
\end{equation}
Let $S\in\{L,U\}$ and $\phi_k\triangleq (M_k!)^{-\frac{1}{M_k}}$, and we have
\begin{align}\label{O_bound_alpha}
  \mathcal{O}^{int,S}_{k}=&
  \frac{\pi\lambda_k}{\mathcal{T}_k}\sum_{m=0}^{M_k}
  \binom{M_k}{m}(-1)^{m}\times\nonumber\\
  &\int_{0}^{R_k^2}
  e^{-\pi \sum_{j=1}^{K}\lambda_j
\Omega_{j,k}^{{\delta_{j}}}\left(1+\mathcal{A}_j
\varphi^{S}_{j,k}(0)\right)
x^{{\delta_{j,k}}}}dx,
\end{align}
where $\varphi^{L}_{j,k}(0)$ and $\varphi^{U}_{j,k}(0)$ have the same forms as $\varphi_{j,k}(0)$ given in \eqref{varphi} with $\beta$ replaced by $m\phi_k\beta$ and $m\beta$, respectively.}

\textit{When $\{\alpha_j\}=\alpha$, we further obtain
\begin{align}\label{O_bound}
  \mathcal{O}^{int,S}_{k} &=
  \frac{\lambda_k}{\mathcal{T}_k}\sum_{m=0}^{M_k}
  \binom{M_k}{m}\frac{(-1)^{m}}
  {\mathcal{Z}_{k}^{S}}
  \left(1-e^{-\pi\mathcal{Z}_{k}^{S}R_k^2}\right),
\end{align}
where $\mathcal{Z}_{k}^{L}$ and $\mathcal{Z}_{k}^{U}$ have the same forms as $\mathcal{Z}_{k}$ given in Theorem \ref{O_k_theorem}  with $\beta$ replaced by $m\phi_k\beta$ and $m\beta$, respectively.}
\end{theorem}
\begin{proof}
    Please see Appendix \ref{appendix_O_k_bound_theorem}.
\end{proof}

Compared to \eqref{O_k_int}, \eqref{O_bound} is much more computationally convenient, requiring only the lookup of ${_2}F_1(\cdot)$ rather than any matrix exponential operation.
When $\{M_j\}=M$, $\{B_j\}=B$, and $\{\alpha_j\}=4$,  $\mathcal{O}^{int,S}_{k}$ reduces to a single trigonometric function with $\mathcal{Z}_{k}^{L}$ and $\mathcal{Z}_{k}^{U}$ sharing the same form as $\mathcal{Z}_{k}$, simply with $\beta$ replaced by $m\phi_k\beta$ and $m\beta$, respectively, i.e.,
\begin{align}
    \begin{cases}\label{Upsilon_L_U}
       \mathcal{Z}_{k}^{L} = \sum_{j=1}^{K}\lambda_j\sqrt{P_{j,k}}
      \left(1+{\mathcal{A}_j\sqrt{m\phi_k\beta}}
      \arctan\sqrt{m\phi_k\beta}\right) ,&\\
       \mathcal{Z}_{k}^{U} = \sum_{j=1}^{K}\lambda_j\sqrt{P_{j,k}}
      \left(1+{\mathcal{A}_j
      \sqrt{m\beta}}\arctan\sqrt{m\beta}\right). &
    \end{cases}
    \end{align}

\subsection{Asymptotic Analysis}
Resorting to the analytical results \eqref{O_k_int} and \eqref{O_bound}, we execute some asymptotic analysis on the outage probability.
In particular, we concern ourselves with the behavior of the outage probability w.r.t. the access threshold $\epsilon$.

\subsubsection{Case $\epsilon\rightarrow 0$}
A small $\epsilon$ regime corresponds to a loose limitation on mobile access, in which network designers tend to increase the cell load by enlarging cell's serving regions.
Given that the number of users is generally much larger than that of BSs, i.e., $\lambda_u\gg\lambda_k$, $\forall k\in\mathcal{K}$, we obtain the following property of the outage probability at the small $\epsilon$ regime.
\begin{property}\label{small_threshold}
   \textit{ At the small $\epsilon$ regime with all tiers sharing the same $M$ and $B$, and $\lambda_u\gg\lambda_l$, $\forall l\in\mathcal{K}$, the outage probability $\mathcal{O}_k^{int}$, $\forall k\in\mathcal{K}$, increases as $\epsilon$ decreases.
   Both $\mathcal{O}_k^{int}$ and the overall $\mathcal{O}^{int}$ converge to the same constant value as $\epsilon\rightarrow 0$, which is independent of $P_l$ and $\lambda_l$.}
\end{property}
\begin{proof}
    When $\lambda_u\gg\lambda_l$, we have $\mathcal{A}_l\rightarrow 1$, $\forall l\in\mathcal{K}$.
    When $M$ and $B$ remain the same for all tiers, $\theta_{k,n}$ in Theorem \ref{O_k_theorem} can be re-expressed as $\theta_{k,n}=\frac{\Lambda}{\Omega_k^{\delta}}
    \varphi(n)$, where $\varphi(n)$ is obtained from \eqref{varphi} by omitting the subscript ``$j,k$" from $\varphi_{j,k}(n)$ and is independent of $P_l$ and $\lambda_l$.
    Accordingly, we have $\mathcal{Z}_k=\theta_{k,0}+\frac{\Lambda}{\Omega_k^{\delta}}
    =\frac{\Lambda}{\Omega_k^{\delta}}(1+\varphi(0))$.
    We define
    $\hat{\mathbf{\Theta}}_{M}
    \triangleq\frac{\Omega_k^{\delta}}{\Lambda}
    \mathbf{\Theta}_{M}$, which is also independent of $P_l$ and $\lambda_l$.
   As a small $\epsilon$ corresponds to a large $R_k$, by omitting the exponential term related to $R_k$ from \eqref{O_k_int} and substituting in $\mathcal{T}_k$, $\mathcal{Z}_k$ and $\mathbf{\Theta}_{M}$, we obtain
    \begin{equation}\label{O_k_small_e}
      \mathcal{O}_k^{int}=1-\frac{1}{1-e^{-\pi\Lambda \epsilon^{-\delta}}}\sum_{i=0}^{M-1}
      \frac{1}{(1+\varphi(0))^{i+1}}
      \left\|\hat{\mathbf{\Theta}}_{M}^i\right\|_1.
    \end{equation}
    We find that $\mathcal{O}_k^{int}$ monotonically increases as $\epsilon$ decreases, and both $\mathcal{O}_k^{int}$ and $\mathcal{O}^{int}$ converge to the following value as $\epsilon\rightarrow 0$,
      \begin{align}\label{O_small_e}
      \!\!\!\mathcal{O}^{int}=\sum_{k=1}^{K}
      \mathcal{T}_k\mathcal{O}_k^{int}
      =\mathcal{O}_k^{int}
      =1-\sum_{i=0}^{M-1}
      \frac{\left\|\hat{\mathbf{\Theta}}_{M}^i
      \right\|_1}{(1+\varphi(0))^{i+1}}.
      \end{align}
       Note that both $\varphi(0)$ and $\hat{\mathbf{\Theta}}_{M}$ are independent of $P_l$ and $\lambda_l$, which completes the proof.
\end{proof}

Property \ref{small_threshold} implies that in order to decrease the outage probability, we should better set a larger $\epsilon$.
From \eqref{O_small_e} we see that the outage probability becomes independent of the BS densities. This indicates that the SIR \emph{invariance} property \cite{Jo2012Heterogeneous}, \cite{Dhillon2012Modeling} still holds in this special case.

\subsubsection{Case $\epsilon\rightarrow \infty$}
A large $\epsilon$ regime corresponds to a rigorous control on mobile access, where network designers are more inclined to improve the downlink quality by decreasing network interference.
In the following property, we show how a large $\epsilon$ influences the outage probability.
\begin{property}\label{large_threshold}
   \textit{ As $\epsilon\rightarrow \infty$,  $\mathcal{O}_k^{int}$ tends to zero, $\forall k\in\mathcal{K}$.}
\end{property}
\begin{proof}
    We prove this property by leveraging the lower and upper bounds on $\mathcal{O}_k^{int}$ given in Theorem \ref{O_k_bound_theorem}.
    As $\epsilon\rightarrow\infty$, we have $\mathcal{A}_k\rightarrow 0$, $\theta_{k,0}\rightarrow 0$, $\forall k\in\mathcal{K}$, and accordingly $\mathcal{Z}_k^{L}=\mathcal{Z}_k^{U}
    =\mathcal{Z}_k
    =\frac{\Lambda}{\Omega_k^{\delta}}$.
     This implies that both the lower and upper bounds converge to the exact value, i.e., $\mathcal{O}_k^{int}=\mathcal{O}_k^{int,L}
    =\mathcal{O}_k^{int,U}$.
    Substituting $\mathcal{Z}_k$ into \eqref{O_bound} yields
    \begin{equation}\label{O_k_proof}
  \mathcal{O}^{int}_{k} =
  \frac{\lambda_k}{\mathcal{T}_k}\sum_{m=0}^{M_k}
  \binom{M_k}{m}\frac{(-1)^{m}\Omega_k^{\delta}}
  {\Lambda}\left(1-e^{-\pi{\Lambda}\epsilon^{-\delta}
  }\right).
\end{equation}
Substituting $\mathcal{T}_k = \frac{\lambda_k\Omega^\delta_k}{\Lambda}
    \left(1-e^{-\pi\Lambda \epsilon^{-\delta}}\right)$ into \eqref{O_k_proof}, $\mathcal{O}^{int}_{k}$ finally reduces to $\sum_{m=0}^{M_k}\binom{M_k}{m}(-1)^{m}=0$, which completes the proof.
\end{proof}

Property \ref{large_threshold} implies that as $\epsilon$ goes to infinity, the outage probability can be reduced to zero.
This also echoes the conclusion drawn from Fig. \ref{OP_W_BETA_EP} that $\mathcal{O}^{int}$ decreases with $\epsilon$.
It is worth mentioning that this asymptotic result does not mean setting $\epsilon$ as large as possible in practice.
As can be seen later in Sec. IV, a properly designed $\epsilon$ is desired to balance the outage probability, spectrum efficiency and energy efficiency.
\subsubsection{Case $P_{1,j}\rightarrow \infty$, $\forall j\in\mathcal{K}\setminus 1$}
The transmit power of the macrocell tier is sufficiently large such that it dominates that of the other tiers.
In this case, we have the following property.

\begin{property}\label{large_P_regime}
	\textit{As $P_{1,j}\rightarrow\infty$, $\forall  j\in\mathcal{K}\setminus 1$, the outage probability $\mathcal{O}_k^{int}$, $\forall k\in\mathcal{K}$, tends to the following value, which depends only on $P_1$ rather than $P_j$,
		\begin{align}\label{O_k_k}
		\mathcal{O}^{int}_k=1-\frac{1}{\mathcal{T}^o_1}\sum_{i=0}^{M_k-1}
		\frac{\left(\mathcal{A}^o_1
			\right)^i
			\left\|\hat{\mathbf{\Theta}}_{M_k}^i\right\|_1}
		{\left[1+\mathcal{A}^o_1\varphi_{1,k}(0)\right]^{i+1}},
		\end{align}
		where $\mathcal{T}^o_1=1 - e^{-\pi\lambda_1\Omega_1^{\delta}\epsilon^{-\delta}}$ and $\mathcal{A}^o_1=1-e^{-\frac{\lambda_u}{\lambda_1}\mathcal{T}^o_1}$.}
\end{property}

\begin{proof}
	As $P_{1,j}\rightarrow\infty$ for $j\in\mathcal{K}\setminus 1$, we have $\mathcal{T}_1\rightarrow \mathcal{T}^o_1=1 - e^{-\pi\lambda_1\Omega_1^{\delta}\epsilon^{-\delta}}$, $\mathcal{T}_j\rightarrow 0$, $\mathcal{A}_1\rightarrow \mathcal{A}^o_1= 1-e^{-\frac{\lambda_u}{\lambda_1}\mathcal{T}^o_1}$, and $\mathcal{A}_j\rightarrow 0$.
	Substituting these results into $\theta_{k,n}$ in Theorem 1, we obtain $\theta_{k,n}=\mathcal{A}^o_1\lambda_1\Omega_{1,k}^{\delta}
	\varphi_{1,k}(n)$ for $k\in\mathcal{K}$.
	Accordingly, we have
	$\mathcal{Z}_k=\left[1+\mathcal{A}^o_1\varphi_{1,k}(0)\right]
	\lambda_1\Omega_{1,k}^{\delta}$ and $\mathbf{\Theta}_{M_k}
	=\mathcal{A}^o_1\lambda_1\Omega_{1,k}^{\delta}
	\hat{\mathbf{\Theta}}_{M_k}$ with $\hat{\mathbf{\Theta}}_{M}$ defined in Property 1.
	In addition, for a sufficiently large $P_1$, $R_1$ and $\Omega_{1,j}$ (or $\mathcal{Z}_j$) become so large that the exponential term can be omitted from (21).
Therefore, the outage probability for the macrocell tier is given by
	\begin{equation}\label{O_k_1}
	\mathcal{O}_1^{int}=1-\frac{1}{\mathcal{T}^o_1}\sum_{i=0}^{M_1-1}
	\frac{\left(\mathcal{A}^o_1
		\right)^i\left\|\hat{\mathbf{\Theta}}_{M_1}^i\right\|_1}
	{\left[1+\mathcal{A}^o_1\varphi_{1,1}(0)\right]^{i+1}}.
	\end{equation}
	The outage probability for tier $j\ne 1$ is calculated as
	\begin{align}\label{O_k_j}
	\mathcal{O}_j^{int}&=1-\frac{\lambda_j}
	{\mathcal{T}_j}\sum_{i=0}^{M_j-1}
	\frac{\left(\mathcal{A}^o_1\lambda_1
		\Omega_{1,j}^{\delta}\right)^i
		\left\|\hat{\mathbf{\Theta}}_{M_j}^i\right\|_1}
	{\left[1+\mathcal{A}^o_1\varphi_{1,j}(0)\right]^{i+1}
		\left(\lambda_1      \Omega_{1,j}^{\delta}\right)^{i+1}}\nonumber\\
	&=1-\sum_{i=0}^{M_j-1}
	\frac{\lambda_j}{\mathcal{T}_j\lambda_1\Omega_{1,j}^{\delta}}
	\frac{(\mathcal{A}^o_1)^i
		\left\|\hat{\mathbf{\Theta}}_{M_j}^i\right\|_1}
	{\left[1+\mathcal{A}^o_1\varphi_{1,j}(0)\right]^{i+1}}
	\nonumber\\
	& \stackrel{\mathrm{(c)}}=1-\frac{1}{\mathcal{T}^o_1}\sum_{i=0}^{M_j-1}
	\frac{(\mathcal{A}^o_1)^i
		\left\|\hat{\mathbf{\Theta}}_{M_j}^i\right\|_1}
	{\left[1+\mathcal{A}^o_1\varphi_{1,j}(0)\right]^{i+1}},
	\end{align}
	where (c) holds for $\lim_{P_{1,j}\rightarrow\infty}
	\frac{\lambda_j}{\mathcal{T}_j\lambda_1\Omega_{1,j}^{\delta}}=
	\frac{\lambda_1\Omega_{1}^{\delta}}
	{\lambda_j\Omega_{j}^{\delta}\left(1 - e^{-\pi\lambda_1\Omega_1^{\delta}\epsilon^{-\delta}}\right)}
	\frac{\lambda_j}{\lambda_1\Omega_{1,j}^{\delta}}=\frac{1}{\mathcal{T}^o_1}$.
	Unifying \eqref{O_k_1} and \eqref{O_k_j} as the form of \eqref{O_k_k} with both $\hat{\mathbf{\Theta}}_{M_k}$ and $\varphi_{1,k}(0)$ independent of $P_j$ (see Property \ref{small_threshold}) completes the proof.
\end{proof}

If $P_1$ continues to increase  to infinity, we have $\mathcal{T}^o_1\rightarrow 1$ and $\mathcal{A}^o_1\rightarrow 1-e^{-\frac{\lambda_u}{\lambda_1}}$, and therefore $\mathcal{O}_k^{int}$ in \eqref{O_k_k} tends to 
\begin{align}\label{O_k_2}
\mathcal{O}^{int}_k=1-\sum_{i=0}^{M_k}
\frac{\left(1-e^{-\frac{\lambda_u}{\lambda_1}}
	\right)^i
	\left\|\hat{\mathbf{\Theta}}_{M_k}^i\right\|_1}
{\left[1+\left(1-e^{-\frac{\lambda_u}{\lambda_1}}
	\right)\varphi_{1,k}(0)\right]^{i+1}},
\end{align}
which becomes independent of both $P_1$ and $\epsilon$.
This implies that the influence of the access threshold on the outage probability becomes less significant as the transmit power of the macrocell tier increases.
We also find that if $\{M_j\}=M$, $\{B_j\}=B$, and $\lambda_u\gg \lambda_l$, $\forall l\in\mathcal{K}$, \eqref{O_k_1} reduces to  \eqref{O_small_e}, i.e.,
all $\mathcal{O}_k$'s converge to the same value.
This indicates that increasing the transmit power may not always benefit the link quality due to the counterbalance of the concurrently increased interference.

\begin{figure}[!t]
	\centering
	\includegraphics[width= 3.0in]{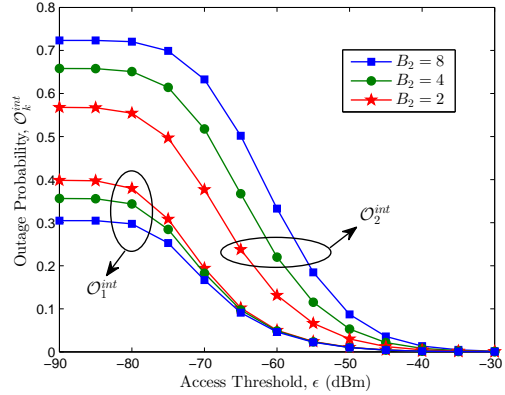}
	\caption{Outage probability vs. $\epsilon$ for different $B_2$'s in a macro/pico HCN, with
		$\{P_1,P_2\}=\{30,10\}$dBm, $\{M_1,M_2\}=\{4,2\}$, $\lambda_1=\frac{1}{\pi 500^2\mathrm{m}^2}$, $\{\lambda_2,\lambda_u\}=\{4\lambda_1, 20\lambda_1\}$, $\beta = 5$dB, and $\{\alpha_j\}=4$.}
	\label{OP_INT_EP_B2}
\end{figure}

\begin{figure}[!t]
	\centering
	\includegraphics[width= 3.0in]{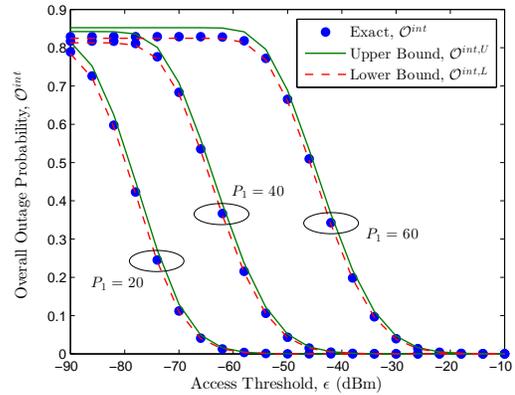}
	\caption{Outage probability vs. $\epsilon$ for different $P_1$'s (dBm) in a macro/pico HCN, with
		$P_2=10$dBm, $\{M_1,M_2\}=\{4,2\}$, $\lambda_1=\frac{1}{\pi 500^2\mathrm{m}^2}$, $\{\lambda_2,\lambda_u\}=\{4\lambda_1, 20\lambda_1\}$, $\beta = 15$dB, and $\{\alpha_j\}=4$.}
	\label{OP_INT_BOUND_EP_P1}
\end{figure}

Fig. \ref{OP_INT_EP_B2} and Fig. \ref{OP_INT_BOUND_EP_P1} depict the outage probability as a function of the access threshold $\epsilon$.
From both figures we see that $\mathcal{O}_k^{int}$ and $\mathcal{O}^{int}$ monotonically decrease with $\epsilon$, and converge to constant values for a small $\epsilon$ while decrease to zero for a large $\epsilon$, which has validated Properties \ref{small_threshold} and \ref{large_threshold}, respectively.

From Fig. \ref{OP_INT_EP_B2}, we see that $\mathcal{O}_1^{int}$ decreases while $\mathcal{O}_2^{int}$ increases as $B_2$ increases.
This implies that although introducing a bias towards admitting users does help a picocell to enlarge the cell load, it deteriorates the outage performance.
The underlying reason is that, as the picocell biasing factor increases, users originally connected to a macrocell but with a low SINR (e.g., at the edge of the macrocell) now become associated with the pococell, which degrades the quality of the picocell link, and in contrast improves the quality of the macrocell link.

From Fig. \ref{OP_INT_BOUND_EP_P1}, we see the lower bound $\mathcal{O}^{int,L}$ is quite close to the exact $\mathcal{O}^{int}$, and the upper bound $\mathcal{O}^{int,U}$ asymptotically approaches to $\mathcal{O}^{int}$ as $\epsilon$ increases.
We also observe that outage probabilities increase with $P_1$, since increasing $P_1$ introduces a large amount of interference into the whole network.
We can further deduce that as $P_1$ increases to infinity $\mathcal{O}^{int}$ tends to an $\epsilon$-independent value, just as indicated by Property \ref{large_P_regime}.

\begin{figure}[!t]
	\centering
	\includegraphics[width= 3.0in]{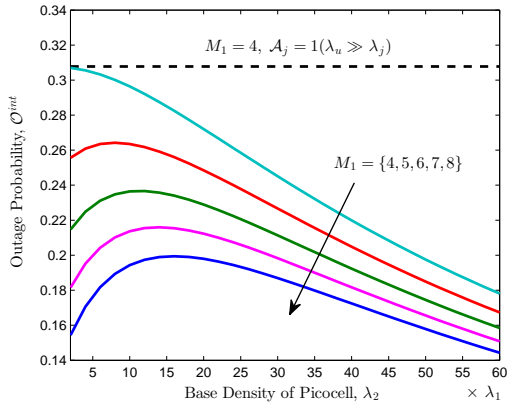}
	\caption{Outage probability vs. $\lambda_2$ for different $M_1$'s in a macro/pico HCN, with
		$\{P_1,P_2\}=\{30,10\}$dBm, $M_2=4$, $\lambda_1=\frac{1}{\pi 500^2\mathrm{m}^2}$, $\lambda_u=50\lambda_1$, $\epsilon=-90$dBm, $\beta = 5$dB, $\{B_j\}=1$, and $\{\alpha_j\}=4$.}
	\label{OP_INT_LA2_MM}
\end{figure}

Fig. \ref{OP_INT_LA2_MM} plots $\mathcal{O}^{int}$ versus $\lambda_2$ for different $M_1$'s.
We see that properly using multiple antennas at BSs efficiently decreases the outage probability.
Suppose all BSs become active, which corresponds to the situation $\lambda_u\gg\lambda_j$, $\forall j\in\mathcal{K}$, then $\mathcal{O}^{int}$ becomes invariant to $\lambda_j$ if $\{M_j\}=M$ and $\{B_j\}=B$.
This verifies the SIR invariance property of an unbiased HCN indicated by Property \ref{small_threshold}.
However for more general scenarios, $\mathcal{O}^{int}$ is heavily influenced by $\lambda_j$.
A general trend is that $\mathcal{O}^{int}$ first increases with $\lambda_2$ and then decreases with $\lambda_2$.
The underlying reason is that, at the small $\lambda_2$ region, $\lambda_u$ dominates $\lambda_2$, keeping the BS activation probabilities large, such that deploying more picocells significantly increases interference and deteriorates the quality of both macrocell links and picocell links.
As $\lambda_2$ increases further, each user connects to a stronger BS whereas interference does not increase significantly since the BS activation probabilities become smaller.

\section{Spectrum Efficiency and Energy Efficiency}
In this section, we first analyze the spectrum efficiency of the HCN by deriving the area network throughput (ANT: $\mathrm{bits/s/Hz/m^2}$) and the average achievable rate (AAR: $\mathrm{bits/s/Hz}$), then we analyze the energy efficiency of the HCN.

\subsection{Area Network Throughput}
We first investigate the ANT in terms of the \emph{transmission capacity} \cite{Weber2005Transmission}, which is defined as
\begin{equation}\label{throughput}
\mathcal{W} = \sum_{k=1}^{K}\lambda_k \mathcal{A}_k(1-\mathcal{O}_{k})\mathcal{R}_{k},
\end{equation}
where $\mathcal{A}_k$, $\mathcal{O}_{k}$ and $\mathcal{R}_{k}=\log_2(1+\beta_k)$ represent the BS activation probability, the outage probability, and the transmission rate of tier $k$, respectively.
SINR thresholds $\beta_k$'s for different tiers are not restricted to be identical.
For convenience, in what follows we just focus on $\mathcal{O}_k^{int}$ instead of $\mathcal{O}_k$.

Previous analytical results and numerical validations show that both $\mathcal{A}_k$ and $\mathcal{O}_k^{int}$ decrease with $\epsilon$.
This implies that $\epsilon$ presents a trade off between the network throughput ($\mathcal{W}$) and the end-to-end reliability ($\mathcal{O}_k$), and only by properly choosing the value of $\epsilon$ can we achieve a good throughput performance.
Generally, $\epsilon$ should be set below $\Omega_k$ for each tier when the propagation loss given in \eqref{LTRP} is taken into consideration.
Then the optimal $\epsilon$ that maximizes $\mathcal{W}$ can be given by
\begin{align}\label{opt_epsilon}
    \epsilon^* = \underset {0<\epsilon<\min_{k\in\mathcal{K}}\Omega_k} {\mathrm{arg~max}}\sum_{k=1}^{K}\lambda_k \mathcal{A}_k
    (1-\mathcal{O}_{k}^{int})\mathcal{R}_{k}.
\end{align}
Substituting $\epsilon^*$ into \eqref{throughput} yields the maximum ANT $\mathcal{W}^*$ .

Unfortunately, since $\mathcal{A}_k$ and $\mathcal{O}_k^{int}$ are coupled by $\epsilon$ in a very complicated way, we can hardly prove the monotonicity or the concavity of $\mathcal{W}$ w.r.t. $\epsilon$, nor can we perform an analytical optimization of $\mathcal{W}$.
The complicated optimization/design problem is out of the scope of this article, instead we concern ourselves with the tractable predictions of network performance and the guidelines for future network design.
In the following properties, we display some insights into the behavior of $\mathcal{W}$ w.r.t. $\epsilon$ leveraging the asymptotic analysis on $\epsilon$.
\begin{property}\label{ANT_small_epsilon}
    \textit{At the small $\epsilon$ regime with $\{M_k\}=M$, $\{B_k\}=B$, and $\lambda_u\gg\lambda_k$, $\forall k\in\mathcal{K}$, $\mathcal{W}$ increases with $\epsilon$.
    As $\epsilon\rightarrow 0$, $\mathcal{W}$ converges to a constant value.}
\end{property}
\begin{proof}
    At the small $\epsilon$ regime with $\lambda_u\gg\lambda_k$, $\forall k\in\mathcal{K}$, we have $\mathcal{A}_k\rightarrow 1$, and \eqref{throughput} can be reduced to $\mathcal{W} = \sum_{k=1}^{K}\lambda_k (1-\mathcal{O}^{int}_{k})\mathcal{R}_{k}$.
    Combined with the asymptotic result of $\mathcal{O}^{int}_{k}$ as $\epsilon\rightarrow 0$ given in \eqref{O_small_e}, we can easily complete the proof.
\end{proof}

\begin{property}\label{ANT_large_epsilon}
    \textit{As $\epsilon\rightarrow \infty$, $\mathcal{W}$ converges to zero.}
\end{property}
\begin{proof}
    Recalling Property \ref{large_threshold}, we have $\mathcal{O}_k^{int}\rightarrow 0$, and $\mathcal{A}_k\rightarrow 0$, $\forall k\in\mathcal{K}$, which directly completes the proof.
\end{proof}

From Properties \ref{ANT_small_epsilon} and \ref{ANT_large_epsilon}, it is intuitive that $\mathcal{W}$ remains unchanged at the small $\epsilon$ region, and as $\epsilon$ increases $\mathcal{W}$ first increases and then decreases, and finally reduces to zero as $\epsilon$ goes to infinity.
Numerical examples are provided in Fig. \ref{ANT_EP_LU} to validate this prediction.
There exists a unique $\epsilon$ that maximizes $\mathcal{W}$, and we find the optimal $\epsilon$ increases with $\lambda_u$.
It is because a larger $\lambda_u$ activates more BSs, introducing more interference to the HCN, hence a larger $\epsilon$ is required to balance the increased interference by decreasing the BS activation probabilities.

\begin{figure}[!t]
\centering
\includegraphics[width= 3.0in]{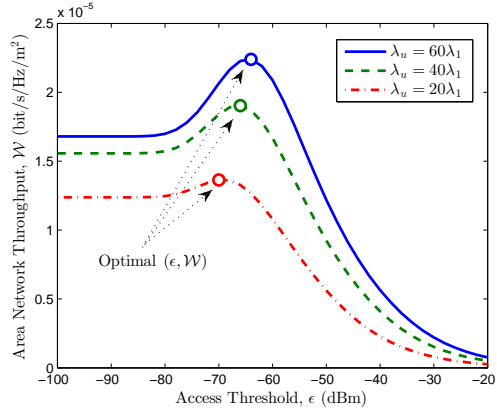}
\caption{Area network throughput vs. $\epsilon$ for different $\lambda_u$'s in a macro/pico HCN, with
$\{P_1,P_2\}=\{30,10\}$dBm, $\{M_1,M_2\}=\{4,2\}$, $\{B_j\}=1$, $\lambda_1=\frac{1}{\pi 500^2\mathrm{m}^2}$, $\lambda_2=10\lambda_1$, $\beta = 10$dB, and $\{\alpha_j\}=4$.}
\label{ANT_EP_LU}
\end{figure}

\begin{figure}[!t]
\centering
\includegraphics[width= 3.0in]{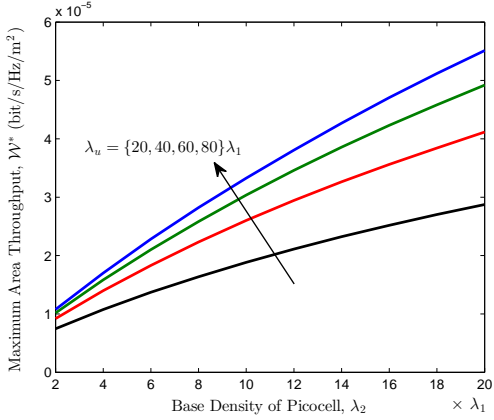}
\caption{Maximum area network throughput vs. $\lambda_2$ for different $\lambda_u$'s in a macro/pico HCN, with
$\{P_1,P_2\}=\{30,10\}$dBm, $\{M_1,M_2\}=\{8,4\}$, $\{B_j\}=1$, $\lambda_1=\frac{1}{\pi 500^2\mathrm{m}^2}$, $\beta = 10$dB, and $\{\alpha_j\}=4$.}
\label{ANT_LA2_LU}
\end{figure}

Fig. \ref{ANT_LA2_LU} depicts the maximum ANT $\mathcal{W}^*$ as a function of $\lambda_2$.
We find that although deploying more picocells may degrade the link quality, it is helpful for improving the area throughput due to the cell densification.
In addition, a larger $\lambda_u$ also increases $\mathcal{W}^*$, since a larger $\lambda_u$ brings more active BSs, which also facilitates the cell densification.

\subsection{Average Achievable Rate}
In this subsection, we evaluate the AAR of a randomly located user, which is given by
\begin{equation}\label{AER_def}
  \mathcal{U} = \sum_{k=1}^K \mathcal{T}_k\mathcal{U}_k,
\end{equation}
where $\mathcal{U}_k$ represents the AAR of the typical user whenever it is associated with tier $k$, which can be calculated as \cite{Jo2012Heterogeneous}
\begin{equation}\label{AER_k_def}
  \mathcal{U}_k = \mathbb{E}_{X_k}\mathbb{E}_{\textsf{{SINR}}_k}
  \left[\log_2(1+\textsf{{SINR}}_k)\right].
\end{equation}
Note that \eqref{AER_k_def} is also the expression of the average ergodic rate in the fast fading channels.
For the most general scenario, we provide an accurate expression of $\mathcal{U}_k$ in the following theorem.
\begin{theorem}\label{AER_W_theorem}
    The AAR of a user associated with tier $k$ is 
    \begin{align}\label{AER_W}
&\mathcal{U}_k= \frac{1}{\ln2}\frac{\lambda_k}{\mathcal{T}_k} \sum_{i=0}^{M_k-1}\sum_{m=0}^{M_k-1}
    \sum_{n=0}^{m}\int_0^{R_k^2}\int_0^{\infty}
    \frac{\pi^{i+1}\sigma_k^{m-n}}
    {(m-n)!i!}\frac{x^{\frac{\alpha_k}{2}(m-n)}}
    {1+\beta}\nonumber\\
    &e^{-\sigma_k x^{\alpha_k/2}-\pi g_0-\pi\sum_{j=1}^{K}\lambda_j
\Omega_{j,k}^{{\delta_{j}}}
x^{{\delta_{j,k}}}}\mathbf{Q}_{M_k}^i(n+1,1)d\beta dx,
\end{align}
where $\sigma_k$, $g_{k,0}$ and $\mathbf{Q}_{M_k}$ have been defined in Theorem \ref{O_k_theorem}.
\end{theorem}
\begin{proof}
    Please see Appendix \ref{appendix_AER_W_theorem}.
\end{proof}

Considering the interference-limited HCN, we provide a more compact expression for $\mathcal{U}_k$ in the following corollary.
\begin{corollary}\label{corollary_ASE_k_int}
    For the interference-limited HCN with $W=0$ and $\{\alpha_j\}=\alpha$, the AAR $\mathcal{U}_k$ in \eqref{AER_k_def} can be given by
    \begin{align}\label{AER_int}
   \mathcal{U}_k^{int}&= \frac{1}{\ln2}\frac{\lambda_k}{\mathcal{T}_k}
     \sum_{i=0}^{M_k-1}\times\nonumber\\
     &\int_0^{\infty}
     \frac{\left\|\mathbf{\Theta}_{M_k}^i\right\|_1}
     {\mathcal{Z}^{i+1}_k}
\left(1 -\sum_{l=0}^i\frac{\pi^{l}
     e^{-\pi\mathcal{Z}_k R_k^2}} {l!~R_k^{-2l}\mathcal{Z}^{-l}_k}\right)
    \frac{d\beta}
    {1+\beta}.
    \end{align}
\end{corollary}
\begin{proof}
    Substituting $\{\alpha_j\}=\alpha$ and $\sigma_k=0$ into \eqref{AER_W} and invoking \eqref{O_k_int}, we complete the proof.
\end{proof}

Although \eqref{AER_int} is not in a closed form, the double integration in \eqref{AER_W} simplifies to a single integration here which can be efficiently calculated as opposed to Monte Carlo simulations.

\begin{figure}[!t]
\centering
\includegraphics[width= 3.0in]{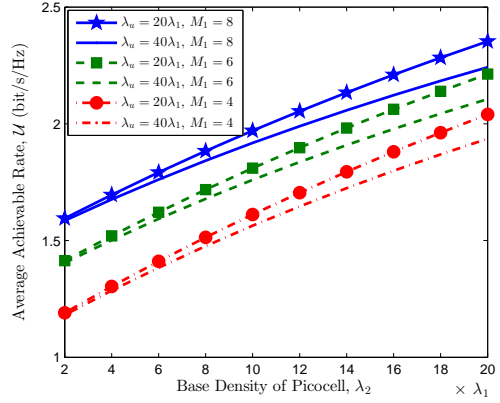}
\caption{Average ergodic rate vs. $\lambda_2$ for different $\lambda_u$'s and $M_1$'s in a macro/pico HCN, with
$\{P_1,P_2\}=\{30,10\}$dBm, $M_2=2$, $\{B_j\}=1$, $\lambda_1=\frac{1}{\pi 500^2\mathrm{m}^2}$, $\epsilon = -60$dBm, and $\{\alpha_j\}=4$.}
\label{AAR_LA2_LU}
\end{figure}
Fig. \ref{AAR_LA2_LU} depicts the AAR as a function of $\lambda_2$.
Similar to Fig. \ref{ANT_LA2_LU}, the AAR increases as $\lambda_2$ increases.
The only difference is that a larger $\lambda_u$ increases the ANT in Fig. \ref{ANT_LA2_LU} whereas decreases the AAR here.
It is because a larger $\lambda_u$ makes more BSs active, introducing more interference to the HCN and degrading the quality of the downlink transmissions.
As a consequence, the AAR per user decreases.
We also see that equipping more antennas at the BS side increases the AAR significantly.

\subsection{Energy Efficiency}
In this subsection, we investigate the energy efficiency of the HCN, which is defined as the ratio of the ANT $\mathcal{W}$ to the area power consumption $\mathcal{P}_A$ \cite{Li2014Throughput}, \cite{Quek2011Energy},
\begin{equation}\label{EE_def}
  \mathcal{F}=\frac{\mathcal{W}}{\mathcal{P}_A}.
\end{equation}
The area power consumption $\mathcal{P}_{A}$  has the following form
\begin{equation}\label{P_total}
  \mathcal{P}_A =\sum_{k=1}^K\lambda_k
  \mathcal{A}_k
  \left(\frac{1}{\eta}P_k+M_kP_C\right)
  +\sum_{k=1}^K\lambda_kP_S,
\end{equation}
where $\eta\in(0,1]$ denotes the linear power amplifier efficiency, $M_kP_C$ accounts for the dynamic circuit power consumption of the BS transmit chains which is proportional to the number of BS antennas, and $P_S$ models the static power consumption in transmit modes which is independent of $M_k$ and $P_k$
 \cite{Ha2013Energy}.

 Note that, as $\epsilon$ increases ($\mathcal{A}_k$ decreases), $\mathcal{W}$ first increases and then decreases (see Fig. \ref{ANT_EP_LU}) while $\mathcal{P}_A$ monotonically decreases, then we may achieve a large $\mathcal{F}$ by properly choosing the value of $\epsilon$.
 Just as validated in Fig. \ref{EE_EP_LU}, $\mathcal{F}$ first increases and then decreases with $\epsilon$, and introducing a proper threshold to mobile access indeed achieves a better energy efficiency than the non-threshold policy does.
 However, since $\mathcal{A}_k(1-\mathcal{O}^{int}_k)$ in the expression of $\mathcal{W}$ is influenced by $\epsilon$ in a very complicated way, it is very difficult to derive the optimal $\epsilon$ that maximizes $\mathcal{F}$.
In the following we are not going to address the optimization problem, instead we give some intuition of the relationship between $\mathcal{F}$ and $\epsilon$ as well as other parameters.

Intuitively, at the small $\epsilon$ regime, as $\mathcal{A}_k$ keeps constant and $\mathcal{W}$ increases with $\epsilon$ (see Property \ref{ANT_small_epsilon}), $\mathcal{F}$ also increases with $\epsilon$.
However, at the large $\epsilon$ regime, $\mathcal{A}_k$ becomes quite small and $\mathcal{P}_A \approx \sum_{k=1}^K\lambda_kP_S$, then $\mathcal{F}$ decreases with $\epsilon$ as $\mathcal{W}$ does (see Property \ref{ANT_large_epsilon}).
 The following propositions give the converged values of $\mathcal{F}$ as $\epsilon$ goes to zero and infinity, respectively.
  \begin{proposition}\label{proposition_F_0}
    As $\epsilon\rightarrow 0$, the energy efficiency $\mathcal{F}$ goes to
\begin{equation}\label{F_0}
  \mathcal{F}^{\circ} = \frac{\sum_{k=1}^{K}\lambda_k \mathcal{A}^{\circ}_k\left(1-\mathcal{O}^{int,\circ}_{k}\right)\mathcal{R}_{k}}
  {\sum_{k=1}^K\lambda_k
  \mathcal{A}^{\circ}_k  \left(\frac{1}{\eta}P_k+M_kP_C\right)
  +\sum_{k=1}^K\lambda_kP_S},
\end{equation}
where $\mathcal{A}_k^{\circ}=1-e^{-
\frac{\lambda_u}{\Lambda}\Omega_k^{\delta}}$, and $\mathcal{O}^{int,\circ}_{k}$ has been given in \eqref{O_k_ep_0}.
  \end{proposition}
  \begin{proof}
    Substituting $\mathcal{A}_k^{\circ}\triangleq \mathcal{A}_k^{\epsilon=0}=1-e^{-
\frac{\lambda_u}{\Lambda}\Omega_k^{\delta}}$ and $\mathcal{O}^{int,\circ}_{k}$ given in \eqref{O_k_ep_0} into \eqref{EE_def} directly completes the proof.
  \end{proof}
  \begin{proposition}
  As $\epsilon\rightarrow \infty$, $\mathcal{F}$ converges to 0.
  \end{proposition}
  \begin{proof}
  Utilizing $\lim_{\epsilon\rightarrow\infty}\mathcal{A}_k=0$ completes the proof.
  \end{proof}

 Fig. \ref{EE_EP_LU} describes $\mathcal{F}$ as a function of $\epsilon$ for different values of $\lambda_u$'s.
 We find that $\mathcal{F}$ keeps constant and is not sensitive to $\lambda_u$ at the small $\epsilon$ regime, whereas increases with $\lambda_u$ observably at the large $\epsilon$ regime.
 The underlying reason is that, when $\epsilon$ is small, $\mathcal{A}_k$ becomes large, and increasing $\lambda_u$ does not increase $\mathcal{A}_k$ too much, such that the increase of $\mathcal{W}$ is counter balanced by the increase of $\mathcal{P}_A$.
 However, when $\epsilon$ becomes large, increasing $\lambda_u$ increases $\mathcal{A}_k$ significantly, and  $\mathcal{W}$ increases faster than $\mathcal{P}_A$ does since the second term of $\mathcal{P}_A$ does not increase with
 $\mathcal{A}_k$.
 Therefore, $\mathcal{F}$ becomes large as $\lambda_u$ increases.

 To focus on evaluating the transmission efficiency of the dynamic power consumed by those active BSs, in what follows we just account for the transmit power and the circuit power consumption at transmit chains \cite{Quek2011Energy}.
  By doing so, the \emph{transmission efficiency} can be expressed as
 \begin{equation}\label{TEE_def}
  \mathcal{F}_T=\frac{\sum_{k=1}^{K}\lambda_k \mathcal{A}_k(1-\mathcal{O}^{int}_{k})\mathcal{R}_{k}}
  {\sum_{k=1}^K\lambda_k
  \mathcal{A}_k  \left(\frac{1}{\eta}P_k+M_kP_C\right)}.
\end{equation}
Similar to Proposition \ref{proposition_F_0}, we infer that $\mathcal{F}_T$ remains unchanged at the small $\epsilon$ regime.
The difference is, as $\epsilon\rightarrow\infty$, $\mathcal{F}$ decreases to zero due to the constant term $P_S$, while the following proposition shows $\mathcal{F}_T$ maintains a constant value.
\begin{proposition}\label{FT_INFTY_corollary}
When $\epsilon\rightarrow \infty$, $\mathcal{F}_T$ converges to
\begin{equation}\label{FT_INFTY}
  \mathcal{F}_T^{\infty}=\frac{\sum_{k=1}^K
  \lambda_k\Omega_k^{\delta}\mathcal{R}_k}
  {\sum_{k=1}^K
  \lambda_k\Omega_k^{\delta}
  \left(\frac{1}{\eta}P_k+M_kP_C\right)}.
\end{equation}
\end{proposition}
\begin{proof}
    Recalling Property \ref{large_threshold}, we have $\lim_{\epsilon\rightarrow\infty}\mathcal{O}^{int}_{k}=0$.
    Utilizing $\lim_{x\rightarrow 0}1-e^{-x}= x$, we have $\lim_{\epsilon\rightarrow\infty}\mathcal{A}_k
    =\frac{\lambda_u}{\Lambda}\Omega_k^{\delta}
    \left(1-e^{-\pi\Lambda\epsilon^{-\delta}}\right)$.
    Substituting these results into \eqref{TEE_def} directly completes the proof.
\end{proof}

Fig. \ref{TEE_EP_LU} plots $\mathcal{F}_T$ as a function of $\epsilon$.
 We see that, different curves of $\mathcal{F}_T$ for different $\lambda_u$'s gradually merge as $\epsilon$ grows, and $\mathcal{F}_T$ becomes independent of $\lambda_u$ as $\epsilon\rightarrow\infty$.
The reason is that, as $\epsilon\rightarrow\infty$, $\mathcal{A}_k$ nearly linearly increases with $\lambda_u$, and so do $\mathcal{W}$ and $\mathcal{P}_A$.
We note that although the cell densification caused by a larger $\lambda_u$ improves network throughput, the improvement is counter balanced by the increased power consumption.
Specifically, when the number of BS antennas is relatively small, an increasing $\lambda_u$ improves $\mathcal{F}_T$, whereas it decreases $\mathcal{F}_T$ when the number of BS antennas is large.
It is because in the latter case, introducing more users dramatically increases the circuit power consumption at transmit chains, which can not be compensated by the improvement of network throughput.
We also find that, although equipping more antennas at BSs is beneficial for improving spectrum efficiency, it turns out to be harmful to energy efficiency due to the rapid increase of circuit power consumption.

\begin{figure}[!t]
\centering
\includegraphics[width= 3.0in]{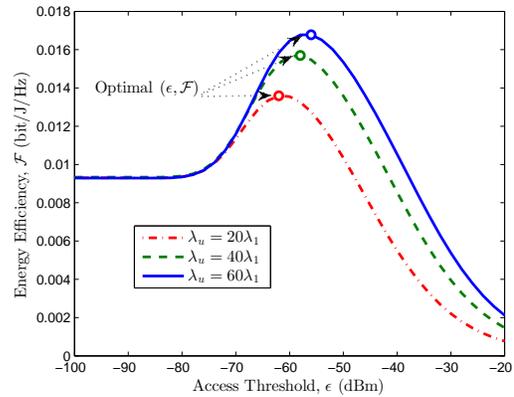}
\caption{Energy Efficiency vs. $\epsilon$ for different $\lambda_u$'s in a macro/pico HCN, with
$\{P_1,P_2,P_C,P_S\}=\{1,0.01,3,4\}\mathrm{watt}$, $\{M_j\}=4$, $\{B_j\}=1$, $\lambda_1=\frac{1}{\pi 500^2\mathrm{m}^2}$, $\lambda_2=10\lambda_1$, $\beta = 10$dB, $\{\alpha_j\}=4$, and $\eta = 0.40$.}
\label{EE_EP_LU}
\end{figure}

\begin{figure}[!t]
\centering
\includegraphics[width= 3.0in]{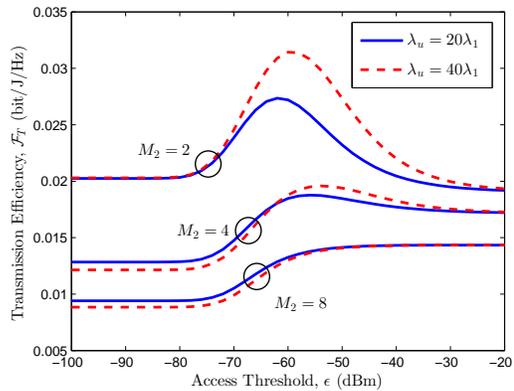}
\caption{Transmission Efficiency vs. $\epsilon$ for different $\lambda_u$'s in a macro/pico HCN, with
$\{P_1,P_2,P_C\}=\{1,0.01,3\}\mathrm{watt}$, $M_1=8$, $\{B_j\}=1$, $\lambda_1=\frac{1}{\pi 500^2\mathrm{m}^2}$, $\lambda_2=10\lambda_1$, $\beta = 10$dB, $\{\alpha_j\}=4$, and $\eta = 0.40$.}
\label{TEE_EP_LU}
\end{figure}

\section{Conclusions}
This paper studies the multi-antenna downlink transmission in a $K$-tier HCN consisting of Poisson random fields of BSs and users.
A reliability-oriented mobile access policy with an access threshold is first proposed, based on which the outage probability of a random user is investigated.
Simple and analytically tractable expressions as well as computationally convenient upper and lower bounds for the outage probabilities are provided to facilitate the analysis.
Spectrum efficiency and energy efficiency are also evaluated from the view of outage.
Theoretic analysis and numerical results show that introducing an access threshold to mobile association not only improves the end-to-end reliability, but also benefits the spectrum and energy efficiencies of the HCN if only the access threshold is properly designed.

\appendix
\subsection{Proof of Lemma \ref{T_k_lemma}}
\label{appendix_T_k_lemma}
We know from \eqref{LTRP} and \eqref{T_k_def} that tier $k$ is chosen only when $D_k\le R_k$ and $ \hat{\rho}_k>\hat{\rho}_j, \forall j\in\mathcal{K}\setminus k$ simultaneously hold.
    Therefore, $\mathcal{T}_k$ can be calculated as
    \begin{equation}\label{T_k_appendix}
        \mathcal{T}_k
        =\int_0^{R_k}\prod_{j\in\mathcal{K}\setminus k} \mathbb{P}\left\{\hat{\rho}_k>\hat{\rho_j}|D_k \right\}f_{D_k}(r)dr,
    \end{equation}
    where $f_{D_k}(r)=2\pi\lambda_k re^{-\pi\lambda_kr^2}$ is the PDF of $D_k$ \cite[Theorem 1]{Haenggi08Distance}.
    $\mathbb{P}\left\{\hat{\rho}_k>\hat{\rho_j}|D_k \right\}$ can be calculated as
   \begin{align}\label{PD_kj}
    &\mathbb{P}
\left\{\hat{\rho}_k>\hat{\rho}_j|D_k \right\}
\stackrel{\mathrm{(d)}} =
        \mathbb{P}\left\{D_j>
        \min\left\{R_j,\Omega_{j,k}^{{\delta_j}/2}
        D_k^{{\delta_{j,k}}}\right\}|D_k\right\}\nonumber\\
       & \stackrel{\mathrm{(e)}}=
        \mathbb{P}\left\{D_j>\Omega_{j,k}^{{\delta_j}/2}
        D_k^{{\delta_{j,k}}}|D_k\right\}\nonumber\\
        &=\mathbb{P}\left\{\textit{No BS in tier j is inside } \mathcal{B}\left(o, \Omega_{j,k}^{{\delta_j/2}}
        D_k^{{\delta_{j,k}}}\right)\right\}\nonumber\\
       & \stackrel{\mathrm{(f)}} = e^{-\pi\lambda_j\Omega_{j,k}^{{\delta_j}}
        D_k^{{2}{\delta_{j,k}}}},
   \end{align}
        where (d) follows from the fact that $\hat{\rho}_k>\hat{\rho}_j$ holds if $D_j>R_j$ or $\frac{P_jM_jB_j}{D_j^{\alpha_j}}
        <\frac{P_kM_kB_k}{D_k^{\alpha_k}}
        \Rightarrow D_j>\Omega_{j,k}^{{\delta_j}/2}
        D_k^{{\delta_{j,k}}}$,
        (e) holds for $R_j = \left(\frac{\Omega_j}{\epsilon}\right)
^{{\delta_j}/{2}} = \Omega_{j,k}
^{{\delta_j/2}}R_k^{\delta_{j,k}}\ge \Omega_{j,k}
^{{\delta_j/2}}D_k^{{\delta_{j,k}}}$, and (f) is obtained from the basic nature of PPP \cite{Andrews2011Tractable}.
    Substituting \eqref{PD_kj} into \eqref{T_k_appendix} yields \eqref{T_k_alpha}.
    When $\{\alpha_j\}=\alpha$, calculating the integral in \eqref{T_k_alpha} directly completes the proof.

\subsection{Proof of Lemma \ref{A_k_lemma}}
\label{appendix_A_k_lemma}
Without loss of generality, we consider a tagged BS $\mathrm{B}_o^{(k)}$.
To calculate $\mathcal{A}_k$, we first calculate its complementary probability $\bar{\mathcal{A}}_k\triangleq 1-\mathcal{A}_k$, which can be calculated as
    \begin{align}\label{A_k_inverse}
        &\bar{\mathcal{A}}_k
        \stackrel{\mathrm{(g)}}=\mathbb{E}_{\Phi_u}\!\!\left[\prod_{x\in
        \Phi_u\bigcap \mathcal{B}(o,R_k)} \mathbb{P}\left\{\textit{$\mathrm{U}_x$ is not associated with $\mathrm{B}_o^{(k)}$}\right\}\right]\nonumber\\
        &=\mathbb{E}_{\Phi_u}\left[\prod_{x\in
        \Phi_u\bigcap \mathcal{B}(o,R_k)} \mathbb{P}\left\{\frac{\Omega_k}{r_{ox}^{\alpha_k}}
        <\max_{j\in\mathcal{K}
        }\hat \rho_j \right\}\right]\nonumber\\
        &\stackrel{\mathrm{(h)}}=\mathbb{E}_{\Phi_u}\left[\prod_{x\in
        \Phi_u\bigcap \mathcal{B}(o,R_k)} 1-e^{-\pi\sum_{j=1}^{K}\lambda_j
\Omega_{j,k}^{{\delta_{j}}}
r_{ox}^{{2}{\delta_{j,k}}}}\right],
    \end{align}
    where (g) measures the probability of an arbitrary user within the serving region of $\mathrm{B}_o^{(k)}$ is associated with the other BSs instead of $\mathrm{B}_o^{(k)}$, and (h) is obtained from \eqref{PD_kj}.
     Using the probability generating functional lemma (PGFL) over PPP \cite{Chiu2016Stochastic} with \eqref{A_k_inverse} easily yields \eqref{A_k_j}.
    When $\{\alpha_j\}=\alpha$, calculating the integral in \eqref{A_k_j} directly completes the proof.

\subsection{Proof of Theorem \ref{O_k_theorem}}
\label{appendix_O_k_theorem}
Recalling \eqref{O_k1}, $\mathcal{O}_{k}$ can be rewritten as
\begin{equation}\label{O_k_app}
  \mathcal{O}_{k}=1-\sum_{m=0}^{M_k-1}\mathbb{E}_{X_k}\left[
  e^{-sW}\sum_{n=0}^{m}\frac{(sW)^{m-n}}{(m-n)!}y_n\right],
\end{equation}
where $y_n \triangleq \frac{(-1)^n s^n}
    {n!}\mathcal{L}^{(n)}_{I_o}(s)$.
     Define $\mathbf{y}_{M}=[y_1,y_2,\cdots,y_{M}]^{\mathrm{T}}$.

The Laplace transform of $I_o$ can be given by
\begin{align}\label{La_Io_app}
    \mathcal{L}_{I_o}(s) = \mathbb{E}_{I_o}\left[e^{-sI_o}\right]
    \stackrel{\mathrm{(i)}}=\prod_{j\in\mathcal{K}} \mathcal{L}_{I_{jo}}(s),
\end{align}
where (i) holds due to
the independence of $I_{io}$ and $I_{jo}$ for $i\neq j$.
 Define $V\triangleq |\mathbf{h}_{zo}^{\mathrm{T}}\mathbf{w}_{z}|^2$.
Using \cite[(8)]{Haenggi2009Stochastic} with \eqref{La_Io_app} yields
\begin{align}\label{La_jo_app}
    &\mathcal{L}_{I_{o}}(s)= \prod_{j\in\mathcal{K}} \mathbb{E}_{\Phi^o_j}\mathbb{E}_V
    \left[\exp\left(-s\sum_{z\in\Phi^o_j\setminus \mathcal{B}(o, r_{j,k})}
    P_jV
    r_{zo}^{-\alpha_j}
    \right)\right]\nonumber\\
    &=\exp\left(-\pi\sum_{j\in\mathcal{K}}
    \lambda_j^o\int^{\infty}_{r_{j,k}^2}
    \left(1-\frac{1}{1+sP_j r^{-\alpha_j/2}}\right)dr\right).
\end{align}

Next, we calculate $\mathcal{L}^{(n)}_{I_o}(s)$ by first presenting it in the following recursive form
\begin{align}\label{p_derivative}
    \mathcal{L}^{(n)}_{I_o}(s)=\pi\sum_{j\in\mathcal{K}}
    &\lambda_j^o\sum_{i=0}^{n-1}\binom{n-1}{i}
    {(-1)^{n-i}}\mathcal{L}^{(i)}_{I_o}(s)\times\nonumber\\
   & \int^{\infty}_{r_{j,k}^2}\frac{(n-i)!\left(
    P_jr^{-{\alpha_j}/{2}}\right)^{n-i}}{\left(
    1+s P_j r^{-{\alpha_j}/{2}}\right)^{n-i+1}}dr.
\end{align}
Employing variable transformation as $r^{-{\alpha_j}/{2}}\rightarrow v$ and substituting $\mathcal{L}^{(n)}_{I_o}(s)$ into $y_n$, we have for $n\ge1$
\begin{align}\label{yp}
    y_n&= \pi\sum_{i=0}^{n-1}
\frac{n-i}{n}y_i\sum_{j\in\mathcal{K}}
{\frac{2\lambda_j^o}{\alpha_j}
   \int_0^{r_{j,k}^{-\alpha_j}}
    \frac{(s P_j)^{n-i}v^{n-i-{\delta_j}-1}}
    {(1+sP_jv)^{n-i+1}}dv}.
\end{align}
Transforming $sP_jv\rightarrow u$ and substituting $s=\frac{\beta X_k^{\alpha_k}}{P_k}$ along with $r_{j,k}=\Omega_{j,k}^{{\delta_j/2}}X_k^{{\delta_{j,k}}}$ into \eqref{yp}, after calculating \eqref{yp} using formula \cite[(3.194.1)]{Gradshteyn2007Table}, we obtain
\begin{equation}\label{xp_2}
  y_n =\pi\sum_{i=0}^{n-1}\frac{n-i}{n}
    g_{k(n-i)}y_i,
\end{equation}
where $g_{kn}\triangleq \sum_{j=1}^K\mathcal{A}_j
\lambda_j\Omega^{\delta_{j}}_{j,k}\varphi_{j,k}(n)
X_k^{2\delta_{j,k}}$, and
\begin{align}\label{varphi_app}
    \varphi_{j,k}(n) =
    \begin{cases}
        \frac{2}{\alpha_j}\tau_{j,k}^{\delta_j}
        \int_0^{\tau_{j,k}}\frac{u^{-\delta_j}}{1+u}du
        , & n = 0,\\
 \frac{2}{\alpha_j}\tau_{j,k}^{\delta_j}
        \int_0^{\tau_{j,k}}\frac{u^{n-\delta_j-1}}
        {(1+u)^{n+1}}du,& n\ne 0 ,
    \end{cases}
\end{align}
with $\tau_{j,k}=\frac{\beta}{M_{j,k}B_{j,k}}$.
According, $y_0=\mathcal{L}_{I_{o}}(s)$ can be derived from \eqref{La_jo_app}, which is $y_0=e^{-\pi g_{k0}}$.

By now, we have obtained a linear recurrence form for $y_n$ in \eqref{xp_2}.
Similar to \cite{Li2014Throughput},
$\mathbf{y}_{M_k}$ has the following matrix form
\begin{equation}\label{xp_matrix}
  \mathbf{y}_{M_k-1}=y_0\sum_{i=1}^{M_k-1}\pi^i
  \mathbf{G}_{M_k-1}^{i-1}\mathbf{g}_{M_k-1},
\end{equation}
where $\mathbf{g}_{M}=[g_{k1},g_{k2},\cdots,g_{kM}]^{\mathrm{T}}$ and
$\mathbf{G}_{M}$ is a lower triangular Toeplitz matrix with its entry $\mathbf{G}_{M}(p,q)=\frac{p-q}{p}g_{k(p-q)}$ for $p>q$ and $\mathbf{G}_{M}(p,q)=0$ for $p\le q$.
From \cite[(39)]{Li2014Throughput}, we have $\mathbf{G}_{M-1}^{i-1}\mathbf{g}_{M-1}
=\frac{1}{i!}\mathbf{Q}_{M}^i(2:M,1)$, where $\mathbf{Q}_{M}^i(2:M,1)$ represent the entries from the second to the $M$-th row in the first column of $\mathbf{Q}_{M}^i$.
Then $y_n$ can be rewritten as
\begin{equation}\label{yp_3}
  y_n = y_0 \sum_{i=0}^{M_k-1}
  \frac{\pi^i}{i!}\mathbf{Q}_{M_k}^i(n+1,1).
\end{equation}
Substituting $y_n$ into \eqref{O_k_app} yields
\begin{align}\label{pc_app_final}
  \mathcal{O}_{k}=1- & \sum_{m=0}^{M_k-1}\sum_{n=0}^{m}\sum_{i=0}^{M_k-1}  \nonumber\\
  &\mathbb{E}_{X_k}\left[
  y_0\frac{\pi^i(sW)^{m-n}e^{-sW}}{(m-n)!i!}
  \mathbf{Q}_{M_k}^i(n+1,1)\right].
\end{align}
To calculate the above expectation, we first give the PDF of $X_k$ in the following lemma by leveraging \cite[Lemma 3]{Jo2012Heterogeneous}.
\begin{lemma}\label{pdf_xk}
\textit{The PDF of $X_k$ is given by
\begin{equation}\label{F_Dk}
  f_{X_k}(x) = \begin{cases}
  ~ \frac{2\pi\lambda_k}{\mathcal{T}_k}x e^{-\pi\sum_{j=1}^{K}\lambda_j
\Omega_{j,k}^{{\delta_{j}}}
x^{{2}{\delta_{j,k}}}}, & x\le R_k,\\
  ~0,& x>R_k.
  \end{cases}
\end{equation}}
\end{lemma}
\begin{proof}
    $f_{X_k}(x)$ for $x\le R_k$ is obtained from  \cite[Lemma 3]{Jo2012Heterogeneous}, while that for $x> R_k$ is an immediate result of \eqref{LTRP}.
\end{proof}

Substituting \eqref{F_Dk} into \eqref{pc_app_final} yields \eqref{O_k_j1}.
When $\{\alpha_j\}=\alpha$, we have $g_{k,n}=\theta_{k,n}X^2_k$ and $\mathbf{Q}_{M}^i=X_k^{2i}\mathbf{\Theta}_{M}^i$.
Resorting to these results, we obtain \eqref{O_k} by calculating the integral in \eqref{O_k_j1}.


\subsection{Proof of Theorem \ref{O_k_bound_theorem}}
\label{appendix_O_k_bound_theorem}
To complete the proof, we first give the following lemma.

\begin{lemma}\label{alzer_lemma}
\textit{(\cite{Alzer1997Mathmatics}):
    Let $q>1$ be a positive real number, and let
    $\phi=[\Gamma(1+1/q)]^{-q}$, then we have for all positive real $\mu$:
\begin{equation}\label{alzer}
  \left(1-e^{-\phi \mu^q}\right)^{1/q}\le
  \frac{1}{\Gamma(1+1/q)}\int_0^\mu e^{-t^q}dt
  \le \left(1-e^{-\mu^q}\right)^{1/q}.
\end{equation}}
\end{lemma}

Since $\|\mathbf{h}_{b}\|^2\sim {\Gamma}(M_k,1)$, its CDF is given by $F_{\|\mathbf{h}_{b}\|^2}(x)
=\int_0^x\frac{e^{-v}v^{M_k-1}}{\Gamma(M_k)}dv$, which can be transformed as $\frac{\int_0^\mu e^{-t^q}dt}{\Gamma(1+1/q)}$ with $p=1/M_k$ and $\mu=x^{M_k}$.
According to Lemma \ref{alzer_lemma}, we have
\begin{equation}\label{cdf_bound}
  \left(1-e^{-\phi_kx}\right)^{M_k}\le F_{\|\mathbf{h}_{b}\|^2}(x)
  \le \left(1-e^{-x}\right)^{M_k},
\end{equation}
where $\phi_k=(M_k!)^{-\frac{1}{M_k}}$.

We first calculate the lower bound of $\mathcal{O}_k^{int}$ as follows.
\begin{align}\label{Ok_lower_app}
\mathcal{O}_{k}^{int}&=\mathbb{E}_{X_k}\mathbb{E}_{I_o} \left[F_{\|\mathbf{h}_{b}\|^2}({s}I_o)\right]
    \ge\mathbb{E}_{X_k}\mathbb{E}_{I_o}\left[
    \left(1-e^{-\phi_{k}sI_o}\right)^{M_k}\right]\nonumber\\
    &=\mathbb{E}_{X_k}\mathbb{E}_{I_o}\left[
    \sum_{n=0}^{M_k}\binom{M_k}{n}(-1)^ne^{-n\phi_{k}
    sI_o}\right]\nonumber\\
    &=\sum_{n=0}^{M_k}\binom{M_k}{n}(-1)^{n}
    \mathbb{E}_{X_k}\left[
   \mathcal{L}_{I_o}(n\phi_{k}s)\right],
\end{align}
where $\mathcal{L}_{I_o}(s)$ has been given in \eqref{La_jo_app}.
Solving the expectation w.r.t. $X_k$ using \eqref{F_Dk} yields the lower bound of $\mathcal{O}_{k}^{int}$.
The upper bound can be easily obtained by replacing $\phi_k$ with 1.

\subsection{Proof of Theorem \ref{AER_W_theorem}}
\label{appendix_AER_W_theorem}
Recalling \eqref{AER_k_def}, $\mathcal{U}_k$ can be calculated as
    \begin{align}\label{Q_k}
     \mathcal{U}_k=   \int_0^{R_k}\mathbb{E}_{\textsf{{SINR}}_k}
  \left[\log_2(1+\textsf{{SINR}}_k(x))\right]f_{X_k}(x)dx,
  \end{align}
  where $f_{X_k}(x)$ is the PDF of $X_k$ given in \eqref{F_Dk}, and $\textsf{{SINR}}_k(x)$ denotes the SINR of a user associated with tier $k$ with distance $x$ away from its serving BS.

 We calculate  $\mathcal{I}=\mathbb{E}_{\textsf{{SINR}}_k}
  \left[\log_2(1+\textsf{{SINR}}_k(X_k))\right]$ as 
  \begin{align}\label{E_S_k}
\mathcal{I}
&=  \int_0^{\infty}\mathbb{P}\left\{
  \log_2(1+\textsf{{SINR}}_k(X_k))>t\right\}dt\nonumber\\
  &\stackrel{\mathrm{(j)}}= \frac{1}{\ln2}\int_0^{\infty}\mathbb{P}\left\{
 \textsf{{SINR}}_k(X_k)>\beta\right\}\frac{1}{1+\beta}d\beta,
  \end{align}
  where (j) holds for the change $t\rightarrow \log_2(1+\beta)$.
  $\mathbb{P}\left\{
  \textsf{{SINR}}_k(X_k)>\beta\right\}$ can be given from \eqref{pc_app_final}
  \begin{align}\label{P_S_k}
 &\mathbb{P}\left\{\textsf{{SINR}}_k(X_k)>\beta\right\}=
 \nonumber\\ &\sum_{m=0}^{M_k-1}\sum_{n=0}^{m}\sum_{i=0}^{M_k-1}
  y_0\frac{\pi^i(sW)^{m-n}e^{-sW}}{(m-n)!i!}
  \mathbf{Q}_{M_k}^i(n+1,1).
  \end{align}
  Substituting \eqref{P_S_k} into \eqref{E_S_k}, and then calculating \eqref{Q_k} using \eqref{F_Dk}, we complete the proof.

\end{document}